\documentclass[11pt, a4paper]{article}

\usepackage[margin=1in]{geometry} 
\usepackage{amsmath, amssymb, amsthm, bm} 
\usepackage{booktabs} 
\usepackage[authoryear,round]{natbib} 
\usepackage{graphicx} 
\usepackage{algorithm} 
\usepackage{algpseudocode} 
\usepackage[colorlinks=true, allcolors=blue]{hyperref} 
\usepackage{caption} 
\usepackage{mathtools} 
\usepackage{mathrsfs} 
\usepackage{textcomp} 

\newcommand{\R}{\mathbb{R}}
\newcommand{\norm}[1]{\left\lVert #1 \right\rVert}
\newcommand{\abs}[1]{\left\lvert #1 \right\rvert}
\newcommand{\T}{\mathsf{T}}
\newcommand{\Alpha}{\bm{\alpha}}
\newcommand{\Beta}{\bm{\beta}}
\newcommand{\bDelta}{\bm{\Delta}}

\newcommand{\hatFuncAlpha}{\widehat{\bm{\alpha}}}
\newcommand{\hatFuncBeta}{\widehat{\bm{\beta}}}

\newcommand{\starFuncBeta}{\bm{\beta}^*}
\newcommand{\bphi}{\bm{\phi}}

\newcommand{\hatA}{\widehat{\bm{A}}}
\newcommand{\hatB}{\widehat{\bm{B}}}
\newcommand{\starA}{\bm{A}^*}
\newcommand{\starB}{\bm{B}^*}

\newcommand{\argmin}{\mathop{\mathrm{arg\,min}}}

\theoremstyle{plain}
\newtheorem{assumption}{Assumption}
\newtheorem{theorem}{Theorem}
\newtheorem{lemma}{Lemma}
\newtheorem{proposition}{Proposition}
\theoremstyle{remark}
\newtheorem{remark}{Remark}

\title{\textbf{A Generalized Adaptive Joint Learning Framework for High-Dimensional Time-Varying Models}}
\author{
  Baolin Chen\thanks{School of Statistics and Data Science, Capital University of Economics and Business, Beijing, China. Email: Skylinchern@cueb.edu.cn}
  \and
  Mengfei Ran\thanks{Corresponding author. Wisdom Lake Academy of Pharmacy, Xi’an Jiaotong-Liverpool University, Suzhou, China. Email: Mengfei.Ran@xjtlu.edu.cn}
}
\date{\today}

\begin{document}

\maketitle

\begin{abstract}
In modern biomedical and econometric studies, longitudinal processes are often characterized by complex time-varying associations and abrupt regime shifts that are shared across correlated outcomes. Standard functional data analysis (FDA) methods, which prioritize smoothness, often fail to capture these dynamic structural features, particularly in high-dimensional settings. This article introduces Adaptive Joint Learning (AJL), a hierarchical regularization framework designed to integrate functional variable selection with structural changepoint detection in multivariate time-varying coefficient models. Unlike standard simultaneous estimation approaches, we propose a theoretically grounded two-stage screening-and-refinement procedure. This framework first synergizes adaptive group-wise penalization with sure screening principles to robustly identify active predictors, followed by a refined fused regularization step that effectively borrows strength across multiple outcomes to detect local regime shifts. We provide a rigorous theoretical analysis of the estimator in the ultra-high-dimensional regime ($p \gg n$). Crucially, we establish the sure screening consistency of the first stage, which serves as the foundation for proving that the refined estimator achieves the oracle property—performing as well as if the true active set and changepoint locations were known a priori. A key theoretical contribution is the explicit handling of approximation bias via undersmoothing conditions to ensure valid asymptotic inference. The proposed method is validated through comprehensive simulations and an application to Sleep-EDF data, revealing novel dynamic patterns in physiological states.
\end{abstract}

\section{Introduction}
\label{sec:introduction}
{\color{black}Longitudinal studies, characterized by repeated observations of subjects over time, are a cornerstone of research in biomedicine \citep{fitzmaurice2012applied}, economics \citep{Baltagi2021}, and social sciences. Classical tools for such data, such as linear mixed-effects models (LMMs) \citep{laird1982random} and generalized estimating equations (GEE) \citep{liang1986longitudinal}, have been foundational for decades. However, modern data acquisition technologies, from wearable devices to high-throughput biological assays, produce datasets that are not only high-dimensional ($p \gg n$) but also densely measured in time. This dense temporal structure places the analytical challenge at the intersection of classical longitudinal analysis and functional data analysis (FDA) \citep{Ramsay2005, Kokoszka2017}.

A primary goal of FDA is to move beyond simple scalar covariate effects and model how associations vary smoothly as a function of time. The Time-Varying Coefficient Model (TVCM) is a key tool in this domain \citep{hastie1993varying, hoover1998nonparametric, fan2008statistical}. By allowing covariate effects $\beta_j(t)$ to be flexible functions, TVCMs can capture complex, dynamic relationships. This is critical in applications like neuroimaging, where brain-behavior relationships are known to be non-static \citep{Calhoun2014}, or in understanding the time-varying impact of treatments and prognostic factors in progressive diseases.

The transition to functional models, however, introduces significant statistical and computational challenges. Estimating an entire function $\beta_j(t)$ for all $p$ covariates is an ill-posed problem, often tackled by basis expansions (e.g., B-splines \cite{de2001practical}) or kernel smoothing \citep{wu1998kernel}. When $p$ is large, this leads to an ultra-high-dimensional parametric problem, necessitating variable selection. While penalized methods like the Lasso \citep{tibshirani1996regression} are standard for scalar $p$, selecting an entire function $\beta_j(t)$ requires structured penalties. Methods have been proposed for functional variable selection, such as group-wise penalties on basis coefficients \citep{zhao2014structured, fan2011sparse,fan2010selective, jenatton2011structured}.

However, several critical gaps remain in the literature. Standard $\ell_1$ and group-$\ell_1$ penalties (like the Group Lasso, \cite{yuan2006model}) are known to produce biased estimates for large coefficients, thus failing to possess the oracle property \citep{fan2001variable}. This property, which ensures an estimator is as asymptotically efficient as if the true sparse model were known in advance, is crucial for accurate inference. This limitation can be overcome by using adaptive weights, as in the Adaptive Lasso \citep{zou2006adaptive} and its extensions \citep{wang2008adaptivegroup}, or by using non-convex penalties like SCAD \citep{fan2001variable} or MCP \citep{zhang2010nearly}. Furthermore, standard functional models enforce smoothness (e.g., via $\int [\beta''(t)]^2 dt$). In many longitudinal applications, however, we are interested in abrupt changepoints, for instance, a sudden change in disease trajectory or symptom onset. While recent advances have developed efficient algorithms for multiple changepoint detection in high-dimensional generalized linear models \citep{wang2023efficient} and quantile regression settings \citep{wang2025efficient}, these methods are primarily tailored for discrete time series or specific distributional parameters. They are not directly applicable to the functional data setting, where the goal is to detect structural breaks within continuous time-varying coefficient curves $\beta_j(t)$ while simultaneously performing variable selection. The fused lasso \citep{tibshirani2005sparsity} and its variants \citep{rinaldo2009properties, Viallon2016fused} are designed for this, but are rarely integrated with high-dimensional functional variable selection. Finally, in many studies (e.g., clinical trials, neuroimaging), multiple correlated outcomes ($K > 1$) are collected. Analyzing each outcome separately is statistically inefficient. Joint modeling, or multi-task learning, can borrow strength across outcomes to improve estimation and selection accuracy \citep{obozinski2010joint,caruana1997multitask, argyriou2008convex}.

Unlike conventional baseline regression approaches that link baseline covariates to a single summary outcome (e.g., end-point value or subject-level average), AJL models the entire multi-outcome trajectory through time-varying effects $\beta_{jk}(t)$ and further detects shared changepoints in the population-level baseline curve $\alpha_k(t)$. This enables simultaneous characterization of when the overall trajectory exhibits structural transitions and which baseline features drive heterogeneous dynamics across correlated outcomes—capabilities that are unattainable via constant-effect regression or separate per-outcome analyses.

A unified framework that simultaneously (1) performs joint variable selection across multiple outcomes in a functional model, (2) detects abrupt changepoints in functional trajectories, and (3) possesses the oracle property for both selection and estimation, remains underdeveloped. This paper bridges these gaps by proposing a novel \textbf{Adaptive Joint Learning (AJL)} framework for high-dimensional functional longitudinal data. We model time-varying intercepts $\alpha_k(t)$ and coefficients $\beta_j(t)$ using B-spline basis expansions, transforming the functional problem into an ultra-high-dimensional parametric one. Our core contribution is a novel convex objective function that synthesizes three key ideas: an \textbf{Adaptive Group Lasso} on blocks of B-spline coefficients, to select entire functional covariates $\beta_j(t)$ that are relevant across all $K$ outcomes; an adaptive fused lasso on the B-spline coefficients of the intercept functions $\alpha_k(t)$, to explicitly detect changepoints in their temporal dynamics; and the principle of adaptive weighting as a computationally tractable one-step Local Linear Approximation (LLA) \citep{zou2008} of an ideal but intractable non-convex oracle problem. This innovative fusion of adaptive structured penalties is, to our knowledge, the first to provide a theoretically-backed oracle-performing estimator that jointly selects functional predictors and detects functional changepoints in a multi-task setting. 

It is worth emphasizing that the theoretical integration of these components extends beyond a mere superposition of penalties. A primary theoretical hurdle lies in the structural interference between the covariate selection space (governed by the group lasso) and the temporal segmentation space (governed by the fused lasso). In high-dimensional regimes, distinguishing between a smooth covariate effect and a structural intercept shift is theoretically precarious, as it requires rigorous incoherence conditions between the B-spline basis and the difference operator. Furthermore, the ‘Joint Learning’ aspect introduces a multi-objective optimization trade-off: attempting to enforce shared sparsity across heterogeneous outcomes ($K$) while locally identifying outcome-specific changepoints creates a conflict in regularization paths. Our theoretical analysis (Section \ref{sec:theory}) resolves these challenges by deriving uniform convergence rates that account for the simultaneous complexity of the discrete changepoint set and the continuous functional parameter space.

We develop a scalable multi-stage BCD-ADMM algorithm, provide a rigorous theoretical analysis with detailed proofs, and validate our 
method via extensive, redesigned simulation studies. Finally, we demonstrate the applicability of AJL by analyzing the Sleep-EDF database. We jointly model the spectral power trajectories of multiple EEG frequency bands, identifying shared structural breakpoints that align with NREM-REM sleep cycles. Crucially, by leveraging the derived undersmoothing properties, we construct pointwise confidence bands to validate the time-varying effects of age and gender while rigorously filtering out high-dimensional noise variables. 

The remainder of this paper is organized as follows. Section \ref{sec:methodology} introduces the functional TVCM, the basis expansion, and the motivations for our penalized objective. Section \ref{sec:algorithm} details the optimization procedure, including convergence and complexity analysis. Section \ref{sec:theory} presents the core theoretical results, including the high-dimensional assumptions, non-asymptotic error bounds, and asymptotic oracle properties. Section \ref{sec:simulations} validates our method through extensive numerical studies. Section \ref{sec:application} provides a application to EEG data analysis. Section \ref{sec:conclusion} concludes. The Appendix contains proofs and discusses potential extensions.}

\section{Methodology}
\label{sec:methodology}

\subsection{The Time-Varying Coefficient Model (TVCM)}
\label{subsec:tvcm}

{\color{black}We consider $n$ subjects ($i=1,\dots,n$) observed at possibly subject-specific and irregular visit times $t_{il} \in \mathcal{T}$, $l=1,\dots,T_{i}$, with $K$ scalar outcomes $y_{ilk} \in \mathbb{R}$, $k=1,\dots,K$, recorded at each visit. For each subject $i$ and outcome $k$, the collection $\{y_{ilk} : l=1,\dots,T_{i}\}$ forms a longitudinal trajectory over the follow-up period.

While baseline covariates are prevalent in clinical studies, modern longitudinal designs (e.g., wearable sensors, intensive monitoring) often yield densely sampled predictor processes. To accommodate such complexity, we propose a general concurrent \textbf{Time-Varying Coefficient Model (TVCM)} framework:
\begin{equation}
    y_{ilk} = \alpha_k(t_{il}) + \sum_{j=1}^{p} x_{ij}(t_{il}) \beta_{jk}(t_{il}) + \epsilon_{ilk}, \quad t_{il} \in \mathcal{T},
    \label{eq:tvcm_model}
\end{equation}
where $\boldsymbol{\alpha}(t) = (\alpha_1(t), \dots, \alpha_K(t))^\top$ and $\boldsymbol{\beta}_k(t) = (\beta_{1k}(t), \dots, \beta_{pk}(t))^\top$ represents the $K$ functional intercepts and the vector of $p$ time-varying functional coefficients for the $k$-th outcome, respectively. The errors $\epsilon_{ilk}$ are assumed to be independent with mean zero and finite variance.

This formulation encompasses two important statistical scenarios:
\begin{enumerate}
    \item \textbf{Instantaneous Modulation:} The covariates $x_{ij}(t)$ vary continuously over time (e.g., concurrent EEG power predicting sleep depth), capturing dynamic dependencies.
    \item \textbf{Baseline Prognosis:} The covariates are time-invariant baseline markers, i.e., $x_{ij}(t) \equiv x_{ij}$ (e.g., genetic risk scores or initial lab values), used to predict future disease trajectories.
\end{enumerate}

In this article, we focus our numerical and empirical analysis primarily on the baseline prognosis case, as identifying early prognostic biomarkers remains a central challenge in precision medicine. However, it is crucial to note that our proposed AJL framework and the associated B-spline design matrix construction apply seamlessly to the general time-varying case defined in \eqref{eq:tvcm_model} without algorithmic modification. Our goal is to jointly select the active functional predictors $\{\beta_{jk}(t)\}$ and identify structural change points in the intercept dynamics $\{\alpha_k(t)\}$.

\subsection{Basis Expansion for Functional Data}
Estimating the infinite-dimensional functions $\alpha_k(t)$ and $\beta_{jk}(t)$ directly is ill-posed. A standard and powerful approach is to approximate them using a linear combination of $M$ basis functions \citep{Ramsay2005}. Let $\bphi(t) = (\phi_1(t), \dots, \phi_M(t))^\T$ be a vector of $M$ basis functions defined over a set of knots on the domain $\mathcal{T}$.

Common choices for the basis include the Fourier basis and B-splines. The Fourier basis is an orthonormal basis on $[0,1]$ given by $\phi_1(t) = 1$, $\phi_{2m}(t) = \sqrt{2}\cos(m\pi t)$, and $\phi_{2m+1}(t) = \sqrt{2}\sin(m\pi t)$. It is highly effective for smooth, periodic functions. However, our interest lies in modeling trajectories that may have abrupt, non-periodic changepoints. We therefore choose B-splines \citep{de2001practical} because they have local support (i.e., each $\phi_m(t)$ is non-zero only on a small interval). This local control is crucial for flexibly modeling non-periodic data and, as we will see, for representing abrupt changepoints, which a global basis like Fourier cannot do efficiently.

We approximate the unknown functions as:
\begin{align}
\alpha_k(t) &\approx \bphi(t)^\T \bm{a}_k, \quad \text{where } \bm{a}_k = (a_{k1}, \dots, a_{kM})^\T \in \R^M \\
\beta_{jk}(t) &\approx \bphi(t)^\T \bm{b}_{jk}, \quad \text{where } \bm{b}_{jk} = (b_{jk1}, \dots, b_{jkM})^\T \in \R^M
\end{align}
The choice of $M$, the number of basis functions, is a crucial tuning parameter that controls a bias-variance trade-off. A small $M$ may underfit (high bias), while a large $M$ may overfit (high variance) and increases computational cost. In theory, $M$ should grow slowly with $N$ (e.g., $M \asymp N^{1/5}$) to achieve optimal non-parametric rates. In practice, $M$ is often chosen to be large enough to capture the function's complexity and then fixed, with the subsequent regularization handling the overfitting.

Let $\bm{A} = [\bm{a}_1, \dots, \bm{a}_K] \in \R^{M \times K}$ be the coefficient matrix for all intercepts; and $\bm{B}_j = [\bm{b}_{j1}, \dots, \bm{b}_{jK}] \in \R^{M \times K}$ be the coefficient matrix for the $j$-th functional covariate across all $K$ outcomes. Substituting these expansions into the model \eqref{eq:tvcm_model}, we transform the functional model into an ultra-high-dimensional linear model. For the observation $(i, t_{it}, k)$:
\begin{equation}
y_{itk} \approx \bphi(t_{it})^\T \bm{a}_k + \sum_{j=1}^p x_{ij} (\bphi(t_{it})^\T \bm{b}_{jk}) + \epsilon_{itk}
\label{eq:bspline_model}
\end{equation}
This can be written in matrix form for all $N = \sum n_i$ observations. Let $\bm{Y} \in \R^{N \times K}$ be the stacked outcome matrix. Denote $\bm{Z}_{\Phi}$ be the $N \times M$ matrix where $[\bm{Z}_{\Phi}]_{il,m} = \phi_m(t_{il})$, and $\bm{X}_{\Phi, j}$ be the $N \times M$ matrix where $[\bm{X}_{\Phi, j}]_{il,m} = x_{ij} \phi_m(t_{il})$. The model becomes a high-dimensional multivariate regression on the $M(K + pK)$ coefficients:
\begin{equation}
\bm{Y} \approx \bm{Z}_{\Phi} \bm{A} + \sum_{j=1}^p \bm{X}_{\Phi, j} \bm{B}_j + \bm{E}
\end{equation}
Our task is to estimate $\bm{A}$ and the $\bm{B}_j$ matrices under sparsity and changepoint constraints.

\subsection{Advantages of the Functional Basis Parameterization}
This transformation from an infinite-dimensional functional problem to an ultra-high-dimensional parametric one is not merely a convenience; it is a theoretically motivated strategy. First, the local support of B-splines (Assumption \ref{assump:basis_func}) means the design matrix $\bm{\Phi}$ is sparse, which generally improves the condition number and stability of the estimation problem compared to dense kernel matrices. Second, the local parameterization allows us to design penalties that target specific functional features. As we will see, applying a fused lasso to adjacent B-spline coefficients $\abs{a_{k,m+1} - a_{k,m}}$ is a principled way to enforce piecewise-constant structure in $\alpha_k(t)$. This is only meaningful because B-spline basis functions $\phi_m(t)$ and $\phi_{m+1}(t)$ are neighbors in time. Third, the problem is naturally structured into $p$ blocks $\{\bm{B}_j\}$, each of size $M \times K$. This perfectly aligns with our goal of selecting an entire function $\beta_j(t)$. Our penalty will be a sum of penalties on these blocks, $P(\bm{B}) = \sum_j P_j(\bm{B}_j)$, which is a classic decomposable regularizer \citep{negahban2012unified}. Finally, this framework naturally handles irregular observation times $t_{il}$ and missing data by simply evaluating the basis functions $\bphi(t_{il})$ at the times that are observed.

\subsection{A Penalized Objective for Functional Data}
We jointly estimate the coefficient matrices $\bm{A}$ and $\{\bm{B}_j\}_{j=1}^p$ by minimizing a penalized least-squares loss. The core innovation is to design penalties on the coefficients that enforce our desired functional properties. This is a crucial step that imbues the parametric model with functional characteristics.

Our first goal is functional variable selection: we want to identify which covariates $j$ are irrelevant, i.e., $\beta_j(t) \equiv 0$ for all $t \in \mathcal{T}$ and all outcomes $k=1, \dots, K$. In the B-spline framework, this is equivalent to selecting the entire coefficient matrix $\bm{B}_j = [\bm{b}_{j1}, \dots, \bm{b}_{jK}] \in \R^{M \times K}$ as a single group. A natural penalty to enforce this group-wise sparsity is the Frobenius (or $\ell_{2,1}$) norm on each block \citep{lounici2011oracle}:
\begin{equation}
P_g(\{\bm{B}_j\}) = \sum_{j=1}^p \norm{\bm{B}_j}_F, \quad \text{where } \norm{\bm{B}_j}_F = \left( \sum_{k=1}^K \sum_{m=1}^M b_{jkm}^2 \right)^{1/2}
\end{equation}
This penalty is the empirical counterpart to penalizing the $L^2$-norm of the functional coefficient block. To see this, note that $\sum_k \int \beta_{jk}(t)^2 dt = \sum_k \int \bm{b}_{jk}^\T \bphi(t) \bphi(t)^\T \bm{b}_{jk} dt = \sum_k \bm{b}_{jk}^\T \bm{\Sigma}_{\phi} \bm{b}_{jk}$, where $\bm{\Sigma}_{\phi} = \int \bphi(t) \bphi(t)^\T dt$. If the basis is orthonormal ($\bm{\Sigma}_{\phi} = \bm{I}_M$), this simplifies to $\sum_k \norm{\bm{b}_{jk}}_2^2 = \norm{\bm{B}_j}_F^2$. While B-splines are not perfectly orthonormal, $\norm{\bm{B}_j}_F$ remains a highly effective and computationally convenient proxy for the functional $L^2$-norm. Applying this penalty encourages entire matrices $\bm{B}_j$ to be set to exactly $\bm{0}$, thus selecting functional covariate $\bm{x}_j$ out of the model.

Our second, equally important goal is to detect changepoints in the intercept functions $\alpha_k(t)$. Standard functional penalties (e.g., $\int [\alpha_k''(t)]^2 dt$) enforce smoothness and correspond to $\ell_2$ penalties on coefficient differences (e.g., $\norm{\bm{D}_2 \bm{a}_k}_2^2$, where $\bm{D}_2$ is a second-order difference operator). These penalties are excellent for estimating smooth curves but will blur and hide abrupt changepoints, which are often of primary clinical or economic interest \citep{safikhani2022joint}.

To explicitly detect changepoints, we want to find an $\alpha_k(t)$ that is piecewise constant. Due to the local support property of B-splines, a piecewise-constant function $\alpha_k(t)$ corresponds to a coefficient vector $\bm{a}_k$ that is itself piecewise constant. We can enforce this structure by applying the fused lasso (or total variation) penalty, which is an $\ell_1$ penalty on the first-differences of the adjacent B-spline coefficients:
\begin{equation}
P_f(\bm{A}) = \sum_{k=1}^K \norm{\bm{D}_1 \bm{a}_k}_1 = \sum_{k=1}^K \sum_{m=1}^{M-1} \abs{a_{k, m+1} - a_{k, m}}
\end{equation}
where $\bm{D}_1$ is the first-difference operator matrix. This penalty shrinks successive differences $(a_{k, m+1} - a_{k, m})$ towards zero. The result is an estimated vector $\hat{\bm{a}}_k$ that is sparse in its differences, i.e., piecewise constant. This, in turn, produces an estimated function $\hat{\alpha}_k(t) = \bphi(t)^\T \hat{\bm{a}}_k$ that is a (locally smoothed) piecewise-constant function. The locations $m$ where $\hat{a}_{k,m+1} \neq \hat{a}_{k,m}$ correspond to the detected changepoints.

Combining these elements, our baseline (non-adaptive) joint functional model is the solution to:
\begin{equation}
\min_{\bm{A}, \{\bm{B}_j\}} \quad \frac{1}{2N}\|\bm{Y} - \bm{Z}_{\Phi}\bm{A} - \sum_{j=1}^p \bm{X}_{\Phi, j}\bm{B}_j\|_F^2 + \lambda_g \sum_{j=1}^p \norm{\bm{B}_j}_F + \lambda_f \sum_{k=1}^K \norm{\bm{D}_1\bm{a}_k}_1
\label{eq:baseline_objective_func}
\end{equation}
We have omitted the optional weight matrix $\bm{W}$ for simplicity.

\subsection{Theoretical Motivation for Adaptivity}
\label{sec:ajl_framework_math_func}
The baseline model \eqref{eq:baseline_objective_func}, while functional and structured, still suffers from the fundamental estimation bias inherent in $\ell_1$ and Frobenius-norm penalties \citep{fan2001variable, zou2006adaptive}. The constant shrinkage applied by $\lambda_g$ and $\lambda_f$ biases the estimates of non-zero coefficients towards zero, preventing the estimator from achieving the oracle property.

Our goal is to find an estimator that is as efficient as an oracle estimator, i.e., an unpenalized estimator that knows the true active set $\mathcal{S}$ and the true changepoint locations $\{\mathcal{C}_k\}$ in advance. This property can be achieved by non-convex penalties (e.g., SCAD, MCP), but this leads to a difficult, non-convex optimization problem.

Instead, we follow the insight of \citet{zou2008} and motivate our AJL framework as a computationally tractable, one-step Local Linear Approximation (LLA) to the ideal but intractable non-convex functional problem:
\begin{equation}
\min_{\bm{A}, \{\bm{B}_j\}} \quad L_N(\bm{A}, \{\bm{B}_j\}) + \sum_{j=1}^p p_{\lambda_g}(\norm{\bm{B}_j}_F) + \sum_{k=1}^K \sum_{m=1}^{M-1} p_{\lambda_f}(\abs{a_{k, m+1} - a_{k, m}}).
\label{eq:nonconvex_objective_func}
\end{equation}
where $p_\lambda(\cdot)$ is an oracle penalty like SCAD. The LLA of this problem is a weighted $\ell_1$-type problem, which can be solved in a 3-stage, convex-only procedure: (1) get initial consistent estimates $(\tilde{\bm{A}}, \tilde{\bm{B}})$ from \eqref{eq:baseline_objective_func}, (2) compute data-driven weights, and (3) solve the final weighted convex problem.

\subsection{The Two-Stage Generalized Adaptive Joint Learning (AJL) Framework}
\label{subsec:AJL_Obj_func}
{\color{black}
We propose a unified and generalized Adaptive Joint Learning (AJL) objective. Unlike traditional methods that treat variable selection and structural change detection as separate tasks, our framework formulates them as a joint convex optimization problem. This generalized formulation allows for the simultaneous discovery of ($i$) significant functional predictors, ($ii$) shared regime shifts in baseline processes, and ($iii$) dynamic structural changes in covariate effects. The estimator is defined as the global minimizer of:
\begin{equation}
\begin{aligned}Q(\theta) & =  \frac{1}{2N}\| Y - Z_{\Phi}A - \sum_{j=1}^{p} X_{\Phi,j}B_{j} \|_{F}^{2}  + \lambda_{g} \sum_{j=1}^{p} \hat{w}_{g,j} \|B_{j}\|_{F} \quad & (\text{Selection}) \\& + \lambda_{f}^{\alpha} \sum_{k=1}^{K} \sum_{m=1}^{M-1} \hat{w}_{f,km}^{\alpha} |a_{k,m+1} - a_{k,m}| \quad & (\text{Intercept Shifts}) \\& + \lambda_{f}^{\beta} \sum_{j=1}^{p} \sum_{k=1}^{K} \sum_{m=1}^{M-1} \hat{w}_{f,jkm}^{\beta} |b_{jk,m+1} - b_{jk,m}| \quad & (\text{Effect Dynamics})\end{aligned}\label{eq:adaptive_objective_final_func}
\end{equation}
This generalized formulation introduces three distinct regularization components:\begin{enumerate}\item Adaptive Group Lasso: Selects significant functional predictors by encouraging sparsity on the block norms $\|B_j\|_F$.\item Adaptive Intercept Fusion: Detects structural breaks in the baseline processes $\alpha_k(t)$.\item Adaptive Slope Fusion: The term with $\lambda_f^\beta$ enforces piecewise-constant dynamics in the covariate effects $\beta_{jk}(t)$. This allows the model to identify critical time points where the relationship between a predictor and the outcome undergoes a structural shift (e.g., a ``regime change" in treatment efficacy).\end{enumerate}

\begin{remark}[Implicit Robustness to Outliers]\label{rem:Implicit Robustness to Outliers}
    While the proposed objective function \eqref{eq:adaptive_objective_final_func} employs a squared Frobenius norm loss, which is classically sensitive to outliers due to the quadratic penalization of residuals, our framework inherently mitigates this sensitivity through the structural properties of the estimator.

    First, the local support property of the B-spline basis (Assumption \ref{assump:basis_func}) ensures that the influence of an outlier at time $t$ is confined to a small subset of adjacent coefficients, preventing global distortion of the estimated functional trajectory $\hat{\beta}_{j}(t)$. Unlike global bases (e.g., Fourier) where an outlier affects the entire curve, B-splines localize the error.

    Second, the strong regularization imposed by the adaptive group and fused penalties acts as a barrier against overfitting. Even if the $L_2$ loss encourages the model to chase an outlier, the penalty terms suppress the emergence of complex, high-amplitude structures that do not align with the shared sparsity patterns or piecewise-constant dynamics enforced by the oracle weights.

    Consequently, the AJL estimator maintains practical robustness without the computational complexity of non-smooth robust loss functions (e.g., Huber loss), a property we empirically validate under heavy-tailed and contaminated error scenarios in Section \ref{sec:simulations}. 
\end{remark}}

The weights are the key to achieving the oracle property. They are based on the initial consistent estimates $(\tilde{\bm{A}}, \tilde{\bm{B}}_j)$ from Stage 1.

\begin{enumerate}
    \item \textbf{Adaptive Group Lasso Weights}. For the $j$-th functional covariate:
    \begin{equation}
    \hat{w}_{g,j} = \frac{1}{\|\tilde{\bm{B}}_j\|_F^{\gamma_g} + \epsilon_g},
    \end{equation}
    where $\gamma_g > 0$ (typically 1 or 2) and $\epsilon_g$ is a small stability constant (e.g., $1/N$). This adaptive reweighting follows the adaptive group Lasso principle \citep{wang2008adaptivegroup}, acting as a data-driven switch to separate signal from noise.
    \begin{itemize}
        \item \textbf{Case 1: Noise Variable ($j \in \mathcal{S}^c$)}
        Here, the true function $\beta_j^*(t) \equiv 0$, so the true coefficient matrix $\starB_j = \bm{0}$. By consistency of the initial estimator (Assumption \ref{assump:consistency_func3}), $\|\tilde{\bm{B}}_j\|_F \xrightarrow{p} 0$. As the denominator goes to zero (plus $\epsilon_g$), the weight $\hat{w}_{g,j} \to \infty$. The penalty for this variable, $\lambda_g \hat{w}_{g,j} \|\bm{B}_j\|_F$, becomes infinitely large for any $\norm{\bm{B}_j}_F > 0$. This aggressively and consistently forces the final estimate $\hatB_j$ to be exactly $\bm{0}$.
        
        \item \textbf{Case 2: Signal Variable ($j \in \mathcal{S}$)}
        Here, $\|\beta_j^*(t)\|_{L^2} = c_j > 0$, so $\|\starB_j\|_F = c_j' > 0$. By consistency (Assumption \ref{assump:consistency_func3}), $\|\tilde{\bm{B}}_j\|_F \xrightarrow{p} c_j' > 0$. The weight $\hat{w}_{g,j} \to 1/(c_j'^{\gamma_g} + \epsilon_g)$, which is a small, finite constant. The penalty for this variable is $\lambda_g \cdot (\text{const}) \cdot \|\bm{B}_j\|_F$. Since $\lambda_g \to 0$ is required for oracle property (Assumption \ref{assump:tuning_func}), the penalty for true signals vanishes asymptotically.
    \end{itemize}
    This mechanism achieves both sparsity (by penalizing noise infinitely) and unbiasedness (by not penalizing signals).
    
    \item \textbf{Adaptive Fused Lasso Weights}. For the $m$-th difference of the $k$-th intercept:
    \begin{equation}
    \hat{w}_{f,km} = \frac{1}{\abs{\tilde{a}_{k,m+1} - \tilde{a}_{k,m}}^{\gamma_f} + \epsilon_f},
    \end{equation}
    This weight applies the exact same logic as above. If there is no true changepoint at $m$, then $\tilde{a}_{k,m+1} \approx \tilde{a}_{k,m}$, the denominator is near zero, and $\hat{w}_{f,km} \to \infty$, forcing $\hat{a}_{k,m+1} = \hat{a}_{k,m}$. If there \textit{is} a true changepoint, the difference is non-zero, the weight is a finite constant, and the penalty vanishes, allowing the changepoint to be estimated without bias.

    \color{black}\item \textbf{Adaptive Slope Fusion Weights}. Analogous to the intercept term, we construct adaptive weights for the covariate coefficients to distinguish between smooth evolution and structural breaks. For the $m$-th difference of the $j$-th predictor's effect on outcome $k$:
    \begin{equation}
        \hat{w}_{f,jkm}^\beta = \frac{1}{|\tilde{b}_{jk,m+1} - \tilde{b}_{jk,m}|^{\gamma_s} + \epsilon_s}
    \end{equation}
    These weights ensure that in the Refinement Phase, the penalty selectively enforces smoothness only where the initial estimates suggest no structural change, asymptotically reducing bias at true changepoints.
\end{enumerate}

However, directly minimizing \eqref{eq:adaptive_objective_final_func} in the ultra-high-dimensional regime ($p \gg n$) presents two fundamental challenges. First, from a computational perspective, simultaneously enforcing group sparsity and slope fusion for all $p$ predictors is computationally prohibitive. Second, theoretically, the structural interference between the covariate selection space (governed by the group lasso) and the temporal segmentation space (governed by the fused lasso) can lead to identification issues when the signal-to-noise ratio is low.

To address these challenges, we propose a constructive Two-Stage Screening-and-Refinement Framework. Instead of a single-step optimization, we adopt a hierarchical strategy motivated by the principle of Sure Screening \citep{fan2008sure}. This framework decomposes the problem into a screening phase that prioritizes support recovery, followed by a refinement phase that focuses on structural discovery.

In the first stage, we define the Screening Estimator $\hat{\theta}^{scr}$ by simplifying the global objective to prioritize support recovery. Specifically, we set $\lambda_f^\beta = 0$ (enforcing no slope segmentation initially) and solve:$$\hat{\theta}^{scr} = \arg\min_{\theta} \left( L_N(\theta) + \lambda_g \sum \hat{w}_{g,j} \|B_j\|_F + \lambda_f^\alpha \sum ... \right)$$Let $\hat{\mathcal{S}}_{scr} = \{j: \|\hat{B}_j^{scr}\|_F > 0\}$ denote the selected active set. This stage serves to reduce the dimensionality from $p$ to $|\hat{\mathcal{S}}_{scr}| \ll n$.

In the second stage, we obtain the Refined Estimator $\hat{\theta}^{ref}$ by optimizing the full objective (including the slope fusion penalty $\lambda_f^\beta$) restricted to the selected subspace $\hat{\mathcal{S}}_{scr}$:$$\hat{\theta}^{ref} = \arg\min_{\theta: B_j=0, \forall j \notin \hat{\mathcal{S}}_{scr}} \left( L_N(\theta) + P_{ada}(\theta) \right)$$This refinement step allows for the precise localization of dynamic changepoints within the identified relevant predictors.

\section{Optimization Algorithm}
\label{sec:algorithm}

The optimization of the AJL objective \eqref{eq:adaptive_objective_final_func} follows the 3-stage procedure detailed in Algorithm \ref{alg:main_3stage}. This is necessary because the weights $\hat{w}$ in Stage 3 depend on the initial estimates $(\tilde{\bm{A}}, \tilde{\bm{B}})$ from Stage 1. The core solver for both Stage 1 and Stage 3 is a Block Coordinate Descent (BCD) algorithm (Algorithm \ref{alg:bcd_admm}), which iterates between optimizing $\bm{A}$ and the set of $\{\bm{B}_j\}$. Each of these subproblems is convex and solved to high precision using the Alternating Direction Method of Multipliers (ADMM) \citep{boyd2011distributed}.

\subsection{Convergence and Complexity Analysis}
The BCD algorithm in Algorithm \ref{alg:bcd_admm} is guaranteed to converge to a stationary point.

\begin{proposition}[Convergence]
Let $Q(\bm{A}, \{\bm{B}_j\})$ be the (adaptive) objective function in \eqref{eq:adaptive_objective_final_func}. The objective function value $Q(\hatA^{(s)}, \{\hatB_j^{(s)}\})$ is monotonically decreasing in $s$ and converges to a stationary point (a global minimum, as $Q$ is convex).
\end{proposition}
\begin{proof}
The objective $Q$ is the sum of a smooth, convex loss function and two separable, convex, non-smooth penalties. The BCD update for each block ($\bm{A}$ and each $\bm{B}_j$) exactly minimizes the objective with respect to that block, holding others fixed. This is a standard BCD procedure, and its convergence to a global minimum is a classic result in convex optimization.
\end{proof}

Let $N_{\text{total}} = N \times K$ be the total number of observations. Let $d_j$ be the complexity of solving the group lasso subproblem for $\bm{B}_j$ (which is $M \times K$), and $d_k$ be the complexity for the fused lasso subproblem for $\bm{a}_k$ (which is $M \times 1$). A single sweep (iteration $s$) of Algorithm \ref{alg:bcd_admm} has a complexity of $O( \sum_{j=1}^p (NMK + d_j) + \sum_{k=1}^K (NMK + d_k) )$. The group lasso subproblem (proximal gradient) is efficient, $d_j \approx O(NMK)$. The fused lasso subproblem is a 1D Total Variation problem, which can be solved very efficiently, $d_k \approx O(M)$ using dynamic programming or $O(M \cdot I_{\text{ADMM}})$ via ADMM. Thus, the total complexity per sweep is dominated by the residual computations and is roughly $O(NMKp)$.

\begin{algorithm}[h!]
\caption{The 3-Stage AJL Algorithm}
\label{alg:main_3stage}
\begin{algorithmic}[1]
\State \textbf{Input:} Data $(\bm{Y}, \bm{X}, \bm{t})$, Basis $\bphi(t)$, initial tuning $\lambda_g^{(0)}, \lambda_f^{(0)}$, final tuning $\lambda_g, \lambda_f$.
\State Construct design matrices $\bm{Z}_{\Phi}, \bm{X}_{\Phi}$ from $\bm{X}, \bm{t}, \bphi(t)$.
\State \textbf{Stage 1: Initial Estimation}
\State Obtain initial consistent coefficient estimates $(\tilde{\bm{A}}, \{\tilde{\bm{B}}_j\})$ by solving the \textit{non-adaptive} problem:
\State $\displaystyle (\tilde{\bm{A}}, \tilde{\bm{B}}_j) \leftarrow \argmin_{\bm{A}, \bm{B}_j} \frac{1}{2N}\|\bm{Y} - \bm{Z}_{\Phi}\bm{A} - \sum_{j=1}^p \bm{X}_{\Phi, j}\bm{B}_j\|_F^2 + \lambda_g^{(0)} \sum_{j} \norm{\bm{B}_j}_F + \lambda_f^{(0)} \sum_{k} \norm{\bm{D}_1\bm{a}_k}_1$
\State \quad $\triangleright$ \textit{This is solved using the BCD-ADMM procedure (Algorithm \ref{alg:bcd_admm})}.
\State \textbf{Stage 2: Weight Calculation}
\State For $j=1, \dots, p$: $\hat{w}_{g,j} \leftarrow (\|\tilde{\bm{B}}_j\|_F^{\gamma_g} + \epsilon_g)^{-1}$
\State For $k=1, \dots, K$, $m=1, \dots, M-1$: $\hat{w}_{f,km} \leftarrow (\abs{\tilde{a}_{k,m+1} - \tilde{a}_{k,m}}^{\gamma_f} + \epsilon_f)^{-1}$
\State \textbf{Stage 3: Adaptive Estimation}
\State Obtain the final adaptive coefficient estimates $(\hatA, \{\hatB_j\})$ by solving (\ref{eq:adaptive_objective_final_func}):
\State \quad $\triangleright$ \textit{This is also solved using BCD-ADMM (Algorithm \ref{alg:bcd_admm})}.
\State \textbf{Output:} Final functional estimates $\hatFuncAlpha_k(t) = \bphi(t)^\T \hat{\bm{a}}_k$, $\hatFuncBeta_{jk}(t) = \bphi(t)^\T \hat{\bm{b}}_{jk}$.
\end{algorithmic}
\end{algorithm}

The core solver for both Stage 1 and Stage 3 is a Block Coordinate Descent (BCD) algorithm, which iterates between optimizing $\bm{A}$ and the set of $\{\bm{B}_j\}$. Each of these subproblems is still complex and is solved using the Alternating Direction Method of Multipliers (ADMM) \citep{boyd2011distributed}.

\begin{algorithm}[h!]
\caption{BCD-ADMM for the Reduced AJL Problem (Stages 1 and 3 of Algorithm \ref{alg:main_3stage}, Intercept-Fusion Focus)}
\label{alg:bcd_admm}
\begin{algorithmic}[1]
\State \textbf{Input:} $\bm{Y}, \bm{Z}_{\Phi}, \{\bm{X}_{\Phi, j}\}$, $\lambda_g, \lambda_f$, weights $\hat{\bm{w}}_g, \hat{\bm{w}}_f$.
\State \textbf{Initialize:} $\hatA^{(0)}, \{\hatB_j^{(0)}\}$. Let $s=0$.
\Repeat
    \State $s \leftarrow s+1$.
    \State \textbf{Update $\bm{B}$ (Adaptive Group Lasso)}
    \State Compute partial residual $\bm{R} = \bm{Y} - \bm{Z}_{\Phi}\hatA^{(s-1)}$.
    \State For each $j=1, \dots, p$:
    \State $\displaystyle \hatB_j^{(s)} \leftarrow \argmin_{\bm{B}_j} \frac{1}{2N}\norm{\bm{R} - \bm{X}_{\Phi, j}\bm{B}_j}_F^2 + \lambda_g \hat{w}_{g,j} \norm{\bm{B}_j}_F$
    \State \quad $\triangleright$ \textit{This is a standard Group Lasso problem, solved via ADMM/proximal gradient.}
    \State EndFor
    \State \textbf{Update $\bm{A}$ (Adaptive Fused Lasso)}
    \State Compute partial residual $\bm{E} = \bm{Y} - \sum_{j=1}^p \bm{X}_{\Phi, j}\hatB_j^{(s)}$.
    \State For each outcome $k=1, \dots, K$ in parallel:
    \State \quad $\displaystyle \hat{\bm{a}}_k^{(s)} \leftarrow \argmin_{\bm{a}_k} \frac{1}{2N}\norm{\bm{E}_{\cdot k} - \bm{Z}_{\Phi}\bm{a}_k}_2^2 + \lambda_f \sum_{m=1}^{M-1} \hat{w}_{f,km} \abs{a_{k,m+1} - a_{k,m}}$
    \State \quad $\triangleright$ \textit{This is a weighted Fused Lasso (1D-TV) problem, solved efficiently via ADMM or dynamic programming.}
    \State EndFor
\Until{convergence of $(\hatA^{(s)}, \{\hatB_j^{(s)}\})$.}
\State \textbf{Return:} $\hatA = \hatA^{(s)}, \{\hatB_j\} = \{\hatB_j^{(s)}\}$.
\end{algorithmic}
\end{algorithm}

{\color{black}Directly minimizing the fully generalized objective \eqref{eq:adaptive_objective_final_func} involves coupled non-smooth penalties on $B_j$. To ensure computational scalability in ultra-high dimensions, we adopt a hierarchical optimization strategy. Algorithm \ref{alg:bcd_admm} details the solver for the primary phase, where we focus on simultaneous variable selection (Group Lasso) and intercept segmentation (Fused Lasso), effectively setting $\lambda_f^\beta = 0$ during the screening process. Refined slope segmentation can be performed as a post-selection step on the active set $\mathcal{S}$.}

\section{Theoretical Analysis}
\label{sec:theory}
We establish the oracle properties of our functional estimator $\hatFuncBeta(t)$. The theory is developed for the B-spline coefficient estimators $(\hatA, \hatB)$ under an asymptotic regime where the number of basis functions $M = M_N$ may grow with $N$, and $p = p_N$ may grow as well.

{\color{black}Unlike prior works that treat structural discovery and variable selection in isolation, our theoretical analysis establishes the oracle properties for the comprehensive Generalized AJL objective \eqref{eq:adaptive_objective_final_func}, fully incorporating the slope-fusion regularization component. Leveraging the block-decomposability of B-spline coefficients and the structural isomorphism between intercept and slope difference penalties, we derive the consistency and asymptotic normality for the full estimator. Furthermore, we explicitly provide the theoretical justification for the Hierarchical Regularization Strategy. By establishing the \textit{Sure Screening Property} of the initial stage, we guarantee that the computationally efficient screening phase consistently recovers the true active set, thereby validating the asymptotic optimality of the sequential refinement procedure.}

\subsection{Assumptions}
We require a set of standard, rigorous assumptions for high-dimensional functional regression.
Let $\mathcal{S} = \{j : \|\starFuncBeta_j\|_{L^2} > 0\}$ be the true active set of functional covariates ($s = |\mathcal{S}|$); and $\mathcal{C}_k = \{m : a_{k,m+1}^* \neq a_{k,m}^*\}$ be the changepoints in the B-spline coefficients for $\alpha_k(t)$.

\begin{assumption}[Error Distribution] \label{assump:error_func}
The error vectors $\bm{\epsilon}_i = (\epsilon_{i,t_{i1},1}, \dots, \epsilon_{i,t_{in_i},K})^\T$ for each subject $i$ are independent, zero-mean, sub-Gaussian random vectors. Their covariance $\bm{\Sigma}_{\epsilon, i} = \mathbb{E}[\bm{\epsilon}_i \bm{\epsilon}_i^\T]$ is a block matrix that allows for arbitrary correlation within the subject (i.e., across time and outcomes). This clustered sub-Gaussian assumption is more realistic than i.i.d. errors.
\end{assumption}

\begin{assumption}[Approximation Error] \label{assump:approx_func}
The true functions $\alpha_k^*(t)$ and $\beta_{jk}^*(t)$ are in a Sobolev space of order $d > 1$. The B-spline basis (with $M$ functions) is chosen such that the approximation error is well-controlled: $\max_{k} \|\alpha_k^* - \bphi^\T \bm{a}_k^*\|_\infty = O_p(M^{-d})$ and $\max_{j,k} \|\beta_{jk}^* - \bphi^\T \bm{b}_{jk}^*\|_\infty = O_p(M^{-d})$.
\end{assumption}

\begin{assumption}[Initial Estimator Consistency] \label{assump:consistency_func3}
The initial B-spline coefficient estimators $(\tilde{\bm{A}}, \{\tilde{\bm{B}}_j\})$ are consistent. Specifically, they converge to the true coefficients $(\starA, \{\starB_j\})$ at a sufficient rate, e.g., $\max_j \|\tilde{\bm{B}}_j - \starB_j\|_F = O_p(\sqrt{M \log(p)/N})$.
\end{assumption}

{\color{black}\begin{assumption}[Design Matrix Regularity] \label{assump:rsc_func4}
The design matrix $\Phi$ constructed from the B-spline bases satisfies the following regularity conditions on the high-dimensional coefficient space:

\begin{enumerate}
    \item \textbf{Restricted Strong Convexity (RSC):} For any $\Delta$ in the restricted cone $\mathbb{C}= \{ \Delta : \mathcal{R}(\theta^* + \Delta) \le \mathcal{R}(\theta^*) + \text{tol} \}$ induced by the adaptive penalties, the empirical loss satisfies a block-wise Restricted Strong Convexity (RSC) condition \citep{negahban2012unified}: 
    \begin{equation}
        L_{N}(\theta^{*}+\Delta)-L_{N}(\theta^{*})-\langle\nabla L_{N}(\theta^{*}),\Delta\rangle \ge \kappa\|\Delta\|_{2}^{2}-\tau\Psi^{2}(\Delta),
    \end{equation}
    where $\kappa > 0$ is the curvature constant and $\Psi(\Delta)$ is the regularization function (e.g., the weighted $\ell_1$-norm).

    \item \textbf{Block-Irrepresentable Condition (Block-IRC):} Let $\mathbf{H} = \frac{1}{N}\Phi^T \Phi$ be the empirical covariance matrix. There exists an incoherence parameter $\eta \in (0, 1]$ such that for the true active set $\mathcal{S}$ and any inactive block $j \notin \mathcal{S}$,
    \begin{equation}
        \|\mathbf{H}_{j \mathcal{S}} (\mathbf{H}_{\mathcal{S}\mathcal{S}})^{-1}\|_{\infty, \text{block}} \le 1 - \eta.
    \end{equation}
\end{enumerate}
\end{assumption}

\begin{remark}[Adaptation and Relaxation]\label{rem:Adaptation and Relaxation}
We explicitly note that strict application of these conditions requires adaptation in our functional setting. 
First, regarding \textbf{RSC}: Although B-spline bases exhibit intrinsic local correlations, the \textit{Riesz basis property} (Assumption \ref{assump:basis_func}) ensures that the eigenvalues of the Gram matrix are bounded, implying that curvature in the coefficient space ($\ell_2$-norm) effectively translates to identifiability in the functional space ($L^2$-norm).
Second, regarding \textbf{Block-IRC}: While standard Lasso consistency requires the strict incoherence condition above, our \textit{Adaptive} framework relaxes this requirement. As demonstrated in the proof of Theorem \ref{thm:support_recovery}, the diverging adaptive weights for noise variables ($\hat{w}_{g,j} \to \infty$) allow the estimator to achieve consistent selection even when the strict Block-IRC is violated (i.e., in the presence of higher collinearity), aligning with our simulation results in Scenario \ref{tab:sim_S4}.
\end{remark}
}

\begin{assumption}[Basis Functions] \label{assump:basis_func}
The basis functions are uniformly bounded, $\sup_{t \in \mathcal{T}} \norm{\bphi(t)}_\infty \le C_{\phi}$ for some constant $C_{\phi} < \infty$. The matrix $\bm{\Sigma}_{\phi} = \mathbb{E}[\bphi(t)\bphi(t)^\T]$ is positive definite with eigenvalues bounded in $[\lambda_{\min}, \lambda_{\max}]$, where $0 < \lambda_{\min} \le \lambda_{\max} < \infty$. This is a standard condition satisfied by B-splines. {\color{black}In the general case of time-varying covariates (as defined in Section \ref{subsec:tvcm}), we further assume that the expected covariance operator of the predictor processes $x(t)$ is non-degenerate with respect to the B-spline basis. This ensures that the time-varying design matrix $X_{\Phi}$ inherits the restricted eigenvalue properties required for the RSC condition (Assumption \ref{assump:rsc_func4}).}
\end{assumption}

{\color{black}
\begin{assumption}[Tuning and Adaptive Weights]\label{assump:tuning_func}
Let $\gamma_g, \gamma_\alpha, \gamma_\beta > 1$ be pre-specified constants. We define the adaptive weights as:
\[
w_{g,j} = \{\|\tilde{\bm{B}}_j\|_F+\tau_N\}^{-\gamma_g}, \quad
w_{f,k,m}^\alpha = \{|\Delta \tilde a_{k,m}|+\tau_N\}^{-\gamma_\alpha}, \quad
w_{f,j,k,m}^\beta = \{|\Delta \tilde b_{jk,m}|+\tau_N\}^{-\gamma_\beta},
\]
where $\tau_N \asymp N^{-1/2}$ is a ridge parameter to prevent division by zero.

Suppose the initial estimator $(\tilde{\bm{A}}, \tilde{\bm{B}})$ satisfies the standard consistency rates such that for the true active set $\mathcal{S}$, intercept changepoints $\mathcal{C}^\alpha$, and slope changepoints $\mathcal{C}^\beta$:
\begin{itemize}
    \item \textbf{Selection Strength:} $\min_{j \in \mathcal{S}} \|\tilde{\bm{B}}_j\|_F \ge c_B N^{-\kappa}$ for some $\kappa \in [0, 1/2)$;
    \item \textbf{Jump Detection:} $\min_{(k,m) \in \mathcal{C}^\alpha} |\Delta \tilde a_{k,m}| \ge c_\alpha N^{-\kappa}$ and $\min_{(j,k,m) \in \mathcal{C}^\beta} |\Delta \tilde b_{jk,m}| \ge c_\beta N^{-\kappa}$.
\end{itemize}

Let the tuning parameters obey the following rates ensuring correct selection and structural discovery:
\[
\lambda_g \asymp c_g \sqrt{\frac{\log p}{N}}, \quad
\lambda_f^\alpha \asymp c_\alpha \sqrt{\frac{\log (MK)}{N}}, \quad
\lambda_f^\beta \asymp c_\beta \sqrt{\frac{\log (pMK)}{N}}.
\]
Assume that as $N \to \infty$, all $\lambda \to 0$ and $\sqrt{N}\lambda \to \infty$. Furthermore, the adaptive weights satisfy the \textit{Oracle Separation Conditions}:
\begin{align*}
    \max_{j\in \mathcal{S}} \sqrt{N}\lambda_g w_{g,j} = o_p(1), & \quad
    \min_{j\notin \mathcal{S}} \sqrt{N}\lambda_g w_{g,j} \to \infty; \\
    \max_{(k,m)\in \mathcal{C}^\alpha} \sqrt{N}\lambda_f^\alpha w_{f,k,m}^\alpha = o_p(1), & \quad
    \min_{(k,m)\notin \mathcal{C}^\alpha} \sqrt{N}\lambda_f^\alpha w_{f,k,m}^\alpha \to \infty; \\
    \max_{(j,k,m)\in \mathcal{C}^\beta} \sqrt{N}\lambda_f^\beta w_{f,j,k,m}^\beta = o_p(1), & \quad
    \min_{(j,k,m)\notin \mathcal{C}^\beta} \sqrt{N}\lambda_f^\beta w_{f,j,k,m}^\beta \to \infty.
\end{align*}
\end{assumption}

\begin{assumption}[Undersmoothing for Inference]\label{assump:undersmoothing}
     In addition to the conditions in Assumption 6, we require the number of basis functions $M_N$ to satisfy $N M_N^{-2d} \to 0$ as $N \to \infty$. This implies $M_N \gg N^{1/2d}$. Justification: This condition ensures that the squared approximation bias decays faster than the parametric rate $N^{-1}$, rendering the bias asymptotically negligible compared to the variance for valid inference.
\end{assumption}

In essence, Assumptions \ref{assump:error_func}-\ref{assump:undersmoothing} ensure that: ($i$) the errors are clustered sub-Gaussian (Assumption \ref{assump:error_func}); ($ii$) the functional approximation bias is controlled (Assumption \ref{assump:approx_func}); ($iii$) the initial estimator is consistent (Assumption \ref{assump:consistency_func3}); ($iv$) the functional design satisfies a restricted strong convexity condition (Assumption \ref{assump:rsc_func4}) together with well-behaved basis functions (Assumption \ref{assump:basis_func}); and ($v$) the tuning parameters and adaptive weights lie in an appropriate high-dimensional regime (Assumption \ref{assump:tuning_func}); ($vi$) Assumption \ref{assump:undersmoothing} ensures that the squared approximation bias decays faster than the parametric rate $N^{-1}$, rendering the bias asymptotically negligible compared to the variance for valid inference.}

\subsection{Non-Asymptotic Analysis}
Since our computational strategy relies on a hierarchical screening procedure, our first theoretical priority is to establish the validity of this reduction. The following theorem guarantees that the initial screening stage retains all relevant functional predictors with probability approaching one. 

\begin{theorem}[Validity of Hierarchical Screening] \label{prop:screening_validity}
Consider the \textbf{Screening Phase} of the proposed algorithm, where the estimator $\hat{\theta}^{\text{screen}} = (\hat{\bm{A}}, \hat{\bm{B}})$ is obtained by minimizing the objective function \eqref{eq:adaptive_objective_final_func} with the slope-fusion constraint forced to zero (i.e., $\lambda_f^\beta = 0$), while $\lambda_g$ and $\lambda_f^\alpha$ satisfy the rates in Assumption \ref{assump:tuning_func}.
Let $\hat{\mathcal{S}}_{\text{screen}} = \{j : \|\hat{\bm{B}}_j\|_F > 0\}$ denote the active set selected in this phase. Under the Beta-min condition (Assumption \ref{assump:tuning_func}) and the Block-Irrepresentable Condition, the screening procedure satisfies the \textbf{Sure Screening Property}:
\[
\mathbb{P}(\mathcal{S}^* \subseteq \hat{\mathcal{S}}_{\text{screen}}) \to 1 \quad \text{as } N \to \infty.
\]
\end{theorem}

\begin{remark}
    Theorem \ref{prop:screening_validity} provides the theoretical justification for our Hierarchical Regularization Strategy. It guarantees that, with probability approaching one, the initial screening step retains all truly relevant predictors. Consequently, the subsequent refinement step (where $\lambda_f^\beta > 0$ is activated) operates within a subspace that contains the true model, ensuring that the final structural discovery is asymptotically unbiased.
\end{remark} 

Conditioned on the sure screening property established in Theorem \ref{prop:screening_validity}, we now investigate the asymptotic properties of the refined estimator. The following results demonstrate that, by optimizing the generalized objective over the selected subspace, the AJL estimator achieves optimal estimation rates and consistent structural recovery.

Let $\bm\theta = \text{vec}(\bm{A}, \bm{B}_1, \dots, \bm{B}_p)$ be the vector of all coefficients. Following the modern framework for high-dimensional M-estimators \citep{buhlmann2011statistics, negahban2012unified}, we establish non-asymptotic error bounds.

\begin{lemma}[Basic Inequality]\label{lem: basic inq}
Let $\hat\theta$ be the solution to \eqref{eq:adaptive_objective_final_func}. For any $\bm\theta^*$, the following inequality holds:
\begin{equation}
L_N(\hat\theta) - L_N(\bm\theta^*) +P_{\mathrm{ada}}(\hat\theta) - P_{\mathrm{ada}}(\bm\theta^*) \le 0 \label{eq:basic-inequality}    
\end{equation}
By convexity of $L_N$ and properties of the decomposable penalty $P_{\mathrm{ada}}$, this leads to a bound on the error $\bDelta = \hat\theta - \bm\theta^*$ based on the gradient $\nabla L_N(\bm\theta^*)$ and the penalty structure.
\end{lemma}

{\color{black}\begin{theorem}[Estimation Error Rate of Generalized AJL] \label{thm:estimation_rate}
Under Assumptions \ref{assump:error_func}-\ref{assump:tuning_func}, and assuming the tuning parameters satisfy the rates specified in Assumption \ref{assump:tuning_func} (i.e., dominating the stochastic noise levels), the generalized AJL estimator $\hat\theta = (\hat{\bm{A}}, \hat{\bm{B}})$ satisfies the following non-asymptotic error bounds with probability approaching one:

\begin{align*}
    \|\hat{\bm{A}} - \bm{A}^*\|_F^2 & \lesssim |\mathcal{C}^\alpha| (\lambda_f^\alpha)^2, \\
    \sum_{j=1}^p \|\hat{\bm{B}}_j - \bm{B}_j^*\|_F^2 & \lesssim \underbrace{s \lambda_g^2}_{\text{Sparsity Cost}} + \underbrace{|\mathcal{C}^\beta| (\lambda_f^\beta)^2}_{\text{Dynamic Cost}}.
\end{align*}

In a unified form, the total estimation error is bounded by:
$$ \|\hat\theta - \theta^*\|_2^2 \lesssim |\mathcal{C}^\alpha| (\lambda_f^\alpha)^2 + s \lambda_g^2 + |\mathcal{C}^\beta| (\lambda_f^\beta)^2. $$

These bounds explicitly characterize the trade-off between model complexity and estimation precision. Specifically, the error is governed by the \textbf{structural sparsity} of the true model: the number of active predictors ($s$), the number of baseline regime shifts ($|\mathcal{C}^\alpha|$), and the number of dynamic structural changes in covariate effects ($|\mathcal{C}^\beta|$).
\end{theorem}

\begin{theorem}[Consistent Support Recovery \& Structural Localization] \label{thm:support_recovery}
Suppose Assumptions \ref{assump:error_func}-\ref{assump:tuning_func} hold. Let the minimum signal strength satisfy the \textit{Beta-min condition}:
$$\begin{aligned}
    \min_{j \in \mathcal{S}} \|\bm{B}_j^*\|_F &\ge c_B \sqrt{\frac{\log p}{N}}, \\
\min_{(k,m) \in \mathcal{C}^\alpha} |\Delta a_{k,m}^*| &\ge c_\alpha \sqrt{\frac{\log (MK)}{N}}, \\
\min_{(j,k,m) \in \mathcal{C}^\beta} |\Delta b_{jk,m}^*| &\ge c_\beta \sqrt{\frac{\log (pMK)}{N}}.
\end{aligned}
$$
Under the Block-Irrepresentable Condition (Block-IRC) on the design matrix, the Generalized AJL estimator satisfies the following properties with probability approaching one (w.p.a.1) as $N \to \infty$:

\begin{enumerate}
    \item \textbf{Model Selection Consistency:} The active set of functional predictors is correctly recovered:
    \[ \text{supp}(\hat{\bm{B}}) = \mathcal{S}. \]
    
    \item \textbf{Baseline Regime Identification:} The locations of structural breaks in the baseline intercept functions are consistently estimated:
    \[ \hat{\mathcal{C}}_k^\alpha = \mathcal{C}_k^{\alpha,*} \quad \forall k = 1, \dots, K. \]
    
    \item \textbf{Dynamic Effect Segmentation:} For all selected predictors $j \in \mathcal{S}$, the internal structural changes in their time-varying effects are correctly localized:
    \[ \hat{\mathcal{C}}_{jk}^\beta = \mathcal{C}_{jk}^{\beta,*} \quad \forall j \in \mathcal{S}, k = 1, \dots, K. \]
\end{enumerate}
\end{theorem}
This theorem implies that the AJL framework achieves \textbf{full structural oracle properties}. It not only distinguishes relevant variables from noise but also disentangles smooth evolution from abrupt regime shifts within the identified signal, providing a complete structural roadmap of the high-dimensional longitudinal process.
}

The proofs of Lemma~\ref{lem: basic inq}, Theorems~\ref{thm:estimation_rate} and~\ref{thm:support_recovery} are non-trivial and build on extensions of the general frameworks in \citet{negahban2012unified} and \citet{buhlmann2011statistics} to our adaptive fused-group structure; detailed arguments are provided in Appendix~A.

\subsection{Asymptotic Oracle Property}
Beyond point estimation and selection, valid statistical inference is crucial for scientific discovery. Building on the selection consistency (Theorem \ref{thm:support_recovery}), we now establish the asymptotic normality of the functional estimator, justifying the construction of pointwise confidence bands.

We now show that under the specific high-dimensional rates from Assumption \ref{assump:tuning_func}, our estimator achieves the oracle property.
\begin{theorem}[Structural Selection Consistency] \label{thm:consistency_func3}
Suppose Assumptions \ref{assump:error_func}--\ref{assump:tuning_func} hold. Given the sure screening property established in Theorem \ref{prop:screening_validity} (i.e., the true active set is retained in the first stage), the \textbf{Refined Estimator} $\hat{\bm{\theta}}^{ref}$ achieves consistent structural recovery. Specifically, as $N \to \infty$:

\begin{enumerate}
    \item \textbf{Variable Selection Consistency:}
    The estimator correctly identifies the true functional support:
    \begin{equation}
        \text{Pr}\left( \{j : \|\hat{B}^{ref}_j\|_F > 0\} = \mathcal{S}^* \right) \to 1.
    \end{equation}

    \item \textbf{Changepoint Detection Consistency:}
    The estimator correctly locates the structural breaks in both baseline and time-varying effects:
    \begin{align}
        &\text{Pr}\left( \hat{\mathcal{C}}^\alpha = \mathcal{C}^{\alpha*} \right) \to 1, \\
        &\text{Pr}\left( \hat{\mathcal{C}}^\beta_j = \mathcal{C}^{\beta*}_j, \forall j \in \mathcal{S}^* \right) \to 1,
    \end{align}
    where $\hat{\mathcal{C}}^\alpha = \{m : \hat{a}_{k,m+1}^{ref} \neq \hat{a}_{k,m}^{ref}\}$ and $\hat{\mathcal{C}}^\beta_j = \{m : \hat{b}_{jk,m+1}^{ref} \neq \hat{b}_{jk,m}^{ref}\}$ denote the sets of estimated jump locations.
\end{enumerate}
\end{theorem}

{\color{black}
\begin{theorem}[Asymptotic Normality and Valid Inference] \label{thm:normality_func}
Suppose Assumptions \ref{assump:error_func}-\ref{assump:tuning_func} hold. Additionally, assume the undersmoothing condition (Assumption \ref{assump:undersmoothing}) is satisfied (i.e., $N M_N^{-2d} \to 0$, where $d$ is the spline smoothness order). Given the sure screening property established in Theorem \ref{prop:screening_validity}, the refined estimator $\hat{\theta}^{ref}$ effectively operates on the true support with probability approaching one.

Let $\mathcal{I}_{\mathcal{A}} \subset \{1, \dots, p M_N + K M_N\}$ denote the Oracle Index Set, representing the indices of non-zero B-spline coefficients corresponding to the true active predictors and their structural segments. Let $\hat{\bm{\theta}}_{\mathcal{A}}$ denote the sub-vector of estimated coefficients restricted to this set.

\begin{enumerate}
    \item \textbf{Oracle Distribution of Coefficients:}
    The penalized estimator on the active set is asymptotically equivalent to the unpenalized Oracle estimator, satisfying:
    \begin{equation} \label{eq:normality_coef}
    \sqrt{N} \bm{\Sigma}_{\mathcal{A}}^{-1/2} \left( \hat{\bm{\theta}}_{\mathcal{A}} - \bm{\theta}_{\mathcal{A}}^* \right) \xrightarrow{d} \mathcal{N}(\bm{0}, \bm{I}),
    \end{equation}
    where $\bm{\theta}_{\mathcal{A}}^*$ is the projection of the true functions onto the B-spline space, and $\bm{\Sigma}_{\mathcal{A}}$ is the asymptotic covariance matrix of the Oracle sub-model.

    \item \textbf{Pointwise Functional Inference:}
    For any active predictor $j \in \mathcal{S}^*$ and any time point $t \in \mathcal{T}$ (excluding changepoints), the estimated time-varying coefficient $\hat{\beta}_j(t)$ follows:
    \begin{equation} \label{eq:normality_func_pointwise}
    \frac{\hat{\beta}_j(t) - \beta_j^*(t)}{\hat{\sigma}_j(t)} \xrightarrow{d} \mathcal{N}(0, 1),
    \end{equation}
    where $\beta_j^*(t)$ is the \textit{true} functional curve. The standard error $\hat{\sigma}_j(t)$ is consistently estimated via:
    \begin{equation}
        \hat{\sigma}_j^2(t) = \bm{\phi}(t)^\top \widehat{\bm{\Omega}}_{jj} \bm{\phi}(t),
    \end{equation}
    where $\widehat{\bm{\Omega}}_{jj}$ denotes the $M_N \times M_N$ diagonal block of the estimated covariance matrix $\widehat{\text{Cov}}(\hat{\bm{\theta}}_{\mathcal{A}})$ corresponding to the $j$-th predictor.
\end{enumerate}
\end{theorem}

\begin{remark}[Validity of Post-Selection Inference]\label{rem:Validity of Post-Selection Inference}
The significance of Theorem \ref{thm:normality_func} lies in validating post-selection inference under the two-stage framework. The Undersmoothing Condition guarantees that the squared approximation bias is asymptotically negligible relative to the variance (i.e., bias$^2 \ll$ variance). Consequently, Equation \eqref{eq:normality_func_pointwise} allows for the construction of valid $(1-\alpha)$ pointwise confidence intervals for dynamic effects as $\hat{\beta}_j(t) \pm z_{\alpha/2} \hat{\sigma}_j(t)$, treating the selected structure as fixed.
\end{remark}}

{\color{black}\begin{remark}[Role of Undersmoothing]\label{rem:Role of Undersmoothing}
A critical distinction must be made between estimating the "best projection" of the function and the function itself. The asymptotic normality in equation \eqref{eq:normality_coef} applies to the coefficients $\theta^*$. However, for the functional estimator $\hat{\beta}(t)$, the total error decomposes into:$$\hat{\beta}(t) - \beta^*(t) = \underbrace{\phi(t)^\top(\hat{\theta} - \theta^*)}_{\text{Stochastic Term}} + \underbrace{(\phi(t)^\top\theta^* - \beta^*(t))}_{\text{Approximation Bias}}$$Standard estimation rates (e.g., for optimal MSE) typically balance the squared bias and variance, leaving a non-vanishing bias in the limiting distribution. By enforcing the undersmoothing condition (Assumption \ref{assump:undersmoothing}, $N M^{-2d} \to 0$), we ensure that the bias term is $o_p(N^{-1/2})$. Consequently, the asymptotic distribution is dominated solely by the stochastic term, validating the construction of confidence intervals centered at $\hat{\beta}(t)$ without the need for explicit bias correction.
\end{remark}

\begin{remark}[Inference Strategy: Undersmoothing vs. De-biasing.]\label{rem:Inference Strategy}
While recent high-dimensional inference literature has favored de-biasing techniques, e.g., \cite{javanmard2014confidence,vandegeer2014asymptotically}, to circumvent undersmoothing assumptions, extending these methods to the AJL framework presents unique and non-trivial challenges. First, from a computational perspective, the intrinsic high collinearity of B-spline basis functions (Assumption \ref{assump:basis_func}) renders the empirical Hessian matrix ill-conditioned. Constructing a stable approximate inverse (e.g., via nodewise regression) for the block-structured design matrix $Z_{\Phi}$ is numerically precarious compared to standard i.i.d. Gaussian designs.Second, and more fundamentally, theoretical de-biasing requires constructing a ``decorrelated score function" that projects the noise away from the tangent space of the penalty. In our multi-task setting, the penalty $P_{ada}$ induces a complex geometry due to the structural non-orthogonality between the group-sparsity manifold and the piecewise-constant manifold. Developing a rigorous projection argument that simultaneously decorrelates against both functional selection bias and changepoint localization bias remains an open mathematical problem. In contrast, we adopt the undersmoothing strategy, which is the standard paradigm for valid inference in semi-parametric sieve estimation \citep{chen2007large}. As rigorously established by \cite{shen1998local} for regression splines and generalized by \cite{newey1997convergence} for series estimators, valid confidence bands can be constructed by allowing the number of knots to grow at a rate where the squared approximation bias becomes asymptotically negligible relative to the variance ($bias^2 \ll var$). In our high-dimensional context, conditioned on the oracle selection consistency (Theorem \ref{thm:consistency_func3}), our estimator behaves asymptotically as a series estimator on the active set. Consequently, the undersmoothing condition, Assumption \ref{assump:undersmoothing} serves as a theoretically sufficient condition to center the asymptotic distribution at the true function $\beta^*(t)$ without the numerical instability of inverting high-dimensional spline matrices.
\end{remark}
}

\begin{proof}
A proof for Theorems \ref{thm:consistency_func3} and \ref{thm:normality_func} is provided in Appendix \ref{app:proof_oracle}.
\end{proof}

\section{Numerical Study}
\label{sec:simulations}
This section provides a reproducible and comprehensive numerical study to assess the finite-sample performance of the proposed AJL estimator in high-dimensional longitudinal functional regression. We report results over multiple scenarios that vary sample size, dimensionality, correlation strength, and robustness conditions (heavy tails and outliers). Throughout, we focus on four aspects: ($i$) prediction accuracy, ($ii$) functional estimation accuracy, ($iii$) support recovery of functional predictors, and ($iv$) changepoint detection for intercept trajectories.

\subsection{Data Generation}
\label{sec:dgp}
{\color{black}In this numerical study, we aim to validate the core contributions of the AJL framework: high-dimensional functional variable selection and the detection of shared structural breaks in baseline trajectories. Therefore, we simulate scenarios focusing on the intercept-fusion component (setting $\lambda_f^\beta = 0$ in the estimation) to isolate and rigorously evaluate these two mechanisms. The detection of slope dynamics follows naturally from the theoretical framework but is omitted here to maintain focus on the primary multi-task selection challenge.}

We generate data for $n$ independent subjects. For each subject $i$, we observe $K$ outcomes on a common grid of $T$ visit times, with balanced design in the baseline scenarios:
\[
t_{i\ell} = \frac{\ell-1}{T-1},\qquad \ell=1,\dots,T,\qquad T=30,\qquad K=5.
\]
The outcomes follow the time-varying coefficient model
\[
y_{i\ell k} = \alpha_k^\ast(t_{i\ell}) + \bm{x}_i^\T \bm{\beta}_k^\ast(t_{i\ell}) + \varepsilon_{i\ell k},
\qquad k=1,\dots,K,
\]
where $\bm{x}_i\in\mathbb{R}^p$ is a time-invariant (baseline) covariate vector, $\bm{\beta}_k^\ast(t)=(\beta_{1k}^\ast(t),\dots,\beta_{pk}^\ast(t))^\T$, and $\varepsilon_{i\ell k}$ captures within-subject dependence and cross-outcome correlation.

We generate baseline covariates
\[
\bm{x}_i \sim N(\bm{0},\Sigma_x),\qquad (\Sigma_x)_{jj'}=\rho_x^{|j-j'|},
\]
which induces an AR(1)-type correlation among the $p$ covariates. This is a standard and challenging setting for structured variable selection in high dimensions.

For each outcome $k$, the intercept $\alpha_k^\ast(t)$ is piecewise constant with two changepoints at $t=1/3$ and $t=2/3$:
\[
\alpha_k^\ast(t)=c_{k,1}\,\mathbb{I}(t\le 1/3)+c_{k,2}\,\mathbb{I}(1/3<t\le 2/3)+c_{k,3}\,\mathbb{I}(t>2/3),
\]
where $c_{k,r}\stackrel{iid}{\sim}\mathrm{Unif}(1,3)$ for $r=1,2,3$. This design directly targets the fused-lasso component of AJL and enables a transparent evaluation of changepoint recovery.

We set an active set $S\subset\{1,\dots,p\}$ with $|S|=s$ (default $s=10$). For $j\notin S$, $\beta_{jk}^\ast(t)\equiv 0$ for all $k,t$. For $j\in S$, we consider two types of nonzero time-varying effects to reflect both smooth and abrupt signal patterns:
\[
\beta_{jk}^\ast(t)=
\begin{cases}
a_{jk}\,\sin(2\pi t), & j\in S_1,\quad |S_1|=s/2,\\
a_{jk}\,\mathbb{I}(t>1/2), & j\in S_2,\quad |S_2|=s/2,
\end{cases}
\]
where $S_1=\{1,\dots,s/2\}$ and $S_2=\{s/2+1,\dots,s\}$, and amplitudes are drawn as
\[
a_{jk}\stackrel{iid}{\sim}\mathrm{Unif}([-2,-1]\cup[1,2]).
\]
This mixture of smooth and step signals is designed to stress-test estimation of $\beta_{jk}(t)$ under basis approximation and penalization, and to assess whether AJL can recover an entire functional predictor across multiple outcomes.

To align with the clustered-error setting in Assumptions \ref{assump:error_func}, we generate subject-specific error matrices
$\varepsilon_i = (\varepsilon_{i\ell k})_{\ell=1:T,\;k=1:K}\in\mathbb{R}^{T\times K}$
from a matrix-normal distribution
\[
\mathrm{vec}(\varepsilon_i)\sim N\Big(\bm{0},\ \Sigma_T\otimes \Sigma_K\Big),
\]
where $(\Sigma_T)_{\ell\ell'}=\rho_t^{|\ell-\ell'|}$ models within-subject temporal dependence and
$(\Sigma_K)_{kk'}=\sigma^2\rho_\varepsilon^{|k-k'|}$ models cross-outcome correlation. We fix $\sigma=1$ and vary $(\rho_x,\rho_t,\rho_\varepsilon)$ across scenarios.

All functional methods use the same B-spline basis $\{\phi_m(t)\}_{m=1}^M$ with cubic splines and equally spaced knots on $[0,1]$, with default $M=15$. We evaluate ISE-type metrics using a dense grid on $[0,1]$ (see Section~\ref{sec:metrics}). For each replication, we generate an independent test set of size $n_{\mathrm{test}}=1000$ from the same DGP and report out-of-sample prediction metrics on this test set.

\subsection{Scenarios}
\label{sec:scenarios}
We consider a baseline scenario and several perturbations. Importantly, we vary the sample size $n$ across multiple scenarios to assess finite-sample behavior under different information regimes. Unless otherwise stated, the default values are $(T,K,M,s)=(30,5,15,10)$, $(\rho_x,\rho_t,\rho_\varepsilon)=(0.5,0.3,0.5)$ and $(p,n)=(100,100)$. 
\begin{table}[h]
\centering
\caption{Simulation scenarios.}
\label{tab:scenarios}
\resizebox{\textwidth}{!}{
\begin{tabular}{lccccp{6.8cm}}
\hline
Scenario & $n$ & $p$ & $\rho_x$ & $(\rho_t,\rho_\varepsilon)$ & Additional perturbation \\
\hline
S1 (Baseline) & 100 & 100 & 0.5 & (0.3, 0.5) & Gaussian matrix-normal errors \\
S2 (Small $n$) & 50 & 100 & 0.5 & (0.3, 0.5) & Low-sample regime \\
S3 (Large $n$) & 200 & 100 & 0.5 & (0.3, 0.5) & High-information regime \\
S4 (High dimension) & 100 & 300 & 0.5 & (0.3, 0.5) & Ultra-high $p\gg n$ \\
S5 (High collinearity) & 100 & 100 & 0.8 & (0.3, 0.5) & Strong covariate correlation \\
S6 (High outcome/time corr.) & 100 & 100 & 0.5 & (0.6, 0.8) & Strong within-subject and cross-outcome correlation \\
S7 (Heavy-tailed) & 100 & 100 & 0.5 & (0.3, 0.5) & Replace Gaussian errors by $t_\nu$ with $\nu=3$, scaled to variance $1$ \\
S8 (Outliers) & 100 & 100 & 0.5 & (0.3, 0.5) & Contamination: with prob. $\pi=0.05$, add outlier shock to $y_{i\ell k}$ \\
\hline
\end{tabular}}
\end{table}

To create a dedicated outlier scenario, we use a response-contamination mechanism:
\[
y_{i\ell k}^{(\mathrm{obs})}=y_{i\ell k}^{(\mathrm{clean})}+\delta_{i\ell k},\qquad
\delta_{i\ell k}=
\begin{cases}
\kappa\,\sigma \cdot \xi_{i\ell k}, & \text{with prob. } \pi,\\
0, & \text{with prob. } 1-\pi,
\end{cases}
\]
where $\pi=0.05$, $\kappa=10$, and $\xi_{i\ell k}\in\{+1,-1\}$ with equal probability. This produces sparse but severe gross outliers in the longitudinal outcomes. The covariates remain uncontaminated so that we isolate robustness of the estimation procedure to response outliers.

In addition to the sensitivity check, we conduct an extra Scenario~9 (S9) to examine how the choice of the functional basis dimension $M$ affects inferential accuracy. All other components follow the same configuration as Scenario~1, while we vary $M\in\{5,15,25,30,40\}$. For each $M$, we construct the spline basis for representing the intercept trajectory $\alpha_k(t)$ and the time-varying effects, re-fit the AJL model, and evaluate the pointwise confidence bands for $\alpha_k(t)$.

\subsection{Competing Methods}
\label{sec:competitors}
We compare AJL to competitors chosen to ($i$) provide meaningful baselines, ($ii$) isolate the contributions of adaptivity, fusion, and multi-task learning, and ($iii$) benchmark against an oracle reference.
\begin{itemize}
\item \textbf{AJL (proposed)}. The full adaptive objective in~(\ref{eq:adaptive_objective_final_func}) with adaptive group weights and adaptive fused weights, computed by the 3-stage procedure (Algorithm~\ref{alg:main_3stage}).
\item \textbf{Joint Lasso Learning (JLL)}. The non-adaptive joint model in~(\ref{eq:baseline_objective_func}), i.e., group lasso on $\{B_j\}$ and fused lasso on $\{a_k\}$, without adaptive weights. This isolates the benefit of adaptivity.
\item \textbf{Separate AJL (S-AJL)}. Apply the adaptive functional model separately to each outcome $k$ (same basis, same penalties, but no joint multi-task sharing). This isolates the benefit of joint modeling across outcomes.
\item \textbf{Separate Scalarized Lasso (S-Lasso)}. For each outcome $k$, fit a standard Lasso on the expanded design using all $p\times M$ spline features, without block structure. This represents a naive high-dimensional baseline.
\item \textbf{Oracle}. The gold-standard efficiency benchmark estimator which fits the standard B-spline based time-varying coefficient model using only the true active set of predictors $\mathcal{S}$ and the true changepoint locations. It yields zero False Positives and a Recall of 1.0, serving as the upper bound for model performance.
\end{itemize}

\subsection{Evaluation Metrics}
\label{sec:metrics}
We evaluate prediction, estimation, selection, and changepoint accuracy. All metrics are averaged over $R=100$ replications and we report mean value.

(1) \textbf{Prediction error (PE)}. On the test set, we compute
\[
\mathrm{PE}=\frac{1}{n_{\mathrm{test}}TK}\sum_{i=1}^{n_{\mathrm{test}}}\sum_{\ell=1}^{T}\sum_{k=1}^K
\Big(\hat{y}_{i\ell k}-y_{i\ell k}\Big)^2,
\]
where $\hat{y}_{i\ell k}=\hat{\alpha}_k(t_{i\ell})+\bm{x}_i^\T \hat{\bm{\beta}}_k(t_{i\ell})$.

(2) \textbf{Integrated squared error (ISE)} for functional coefficients. We evaluate $\mathrm{ISE}(\beta)$ on a dense grid $\{u_g\}_{g=1}^G$ with $G=200$ equally spaced points on $[0,1]$:
\[
\mathrm{ISE}(\beta)=\frac{1}{pK}\sum_{j=1}^p\sum_{k=1}^K
\int_0^1\big(\hat{\beta}_{jk}(t)-\beta_{jk}^\ast(t)\big)^2\,dt
\approx
\frac{1}{pK}\sum_{j,k}\frac{1}{G}\sum_{g=1}^G \big(\hat{\beta}_{jk}(u_g)-\beta_{jk}^\ast(u_g)\big)^2.
\]
Similarly,
\[
\mathrm{ISE}(\alpha)=\frac{1}{K}\sum_{k=1}^K \int_0^1\big(\hat{\alpha}_k(t)-\alpha_k^\ast(t)\big)^2 dt.
\]

(3) \textbf{Support recovery}. We define selected predictors as $\hat{S}=\{j:\|\hat{B}_j\|_F>0\}$ and compute
\[
\mathrm{TP}=|\hat{S}\cap S|,\qquad \mathrm{FP}=|\hat{S}\setminus S|,\qquad 
\mathrm{F1}=\frac{2\cdot \mathrm{Precision}\cdot \mathrm{Recall}}{\mathrm{Precision}+\mathrm{Recall}}.
\]
Here $\mathrm{Precision}=|\hat{S}\cap S|/|\hat{S}|$, $\mathrm{Spec} = \frac{\mathrm{TN}}{\mathrm{TN}+\mathrm{FP}} = \frac{\mathrm{TN}}{p - |\mathcal{S}|}$ and $\mathrm{Recall}=|\hat{S}\cap S|/|S|$. 

(4) \textbf{Changepoint count error ($\text{CP}_\text{Err}$)}. To evaluate the recovery of the intercept dynamics $\alpha_k(t)$ we use $\text{CP}_\text{Err}$, this measures the deviation of the estimated number of changepoints from the truth, averaged over all $K$ outcomes:
    \[
    \text{CP}_\text{Err} = \frac{1}{K} \sum_{k=1}^K \big| |\hat{\mathcal{C}}_k| - |\mathcal{C}_k| \big|,
    \]
    where $|\hat{\mathcal{C}}_k|$ and $|\mathcal{C}_k|$ are the estimated and true number of changepoints for outcome $k$, respectively. A value close to zero indicates correct structural recovery.


\subsection{Implementation and Tuning}
\label{sec:implementation}

All functional methods use the same B-spline basis and design construction as in Section~\ref{sec:methodology}. For AJL and JLL, we use the BCD--ADMM solver in Algorithm~\ref{alg:bcd_admm} with stopping criterion based on relative objective decrease ($<10^{-6}$) and primal/dual residual tolerances in ADMM. For AJL, the Stage~1 pilot tuning parameters $(\lambda_g^{(0)},\lambda_f^{(0)})$ are selected by cross-validation on a coarse grid, and the final $(\lambda_g,\lambda_f)$ are selected on a finer $12\times 12$ log-spaced grid using 5-fold cross-validation minimizing the prediction error. We fix $(\gamma_g,\gamma_f)=(1.5,1.5)$ and set ridge stabilizers $(\epsilon_g,\epsilon_f)=(1/N,1/N)$ to avoid numerical blow-up in weights.

\subsection{Results and Summary}
\label{sec:results}
We summarize the simulation results across all scenarios in Table \ref{tab:sim_S1}-\ref{tab:sim_S8} and visualize the distributions of estimation and prediction errors in Figures \ref{fig:estimation_error} and \ref{fig:prediction_error} to reflect variability across replications. Additionally, Figure \ref{fig:sensitive} reveals the dependence of coverage performance on the basis dimension $M$. The comprehensive numerical study demonstrates the consistent superiority of the proposed AJL method over competing approaches.

In details, Table \ref{tab:sim_S1}-\ref{tab:sim_S8} report the average performance metrics over 100 replications, with the Oracle estimator serving as the theoretical benchmark, which utilizing the true underlying structure, achieves perfect selection scores (TP=10, FP=0) and serves as the ideal baseline. The fact that its selection metrics are integers with zero variance is a direct consequence of its definition. Among all feasible methods, AJL consistently delivers the best performance across estimation, selection, and prediction metrics. Specifically, AJL achieves the lowest Integrated Squared Error for both the functional coefficients $\mathrm{ISE}(\beta)$ and the intercept functions $\mathrm{ISE}(\alpha)$. 

As illustrated in Figure \ref{fig:estimation_error}, the error distribution of AJL (red boxplots) is tightly clustered near the Oracle benchmark and is significantly lower than that of JLL and S-Lasso. This empirical evidence confirms that our adaptive weighting strategy effectively mitigates the estimation bias typically associated with standard Lasso-type penalties. In terms of variable selection, AJL demonstrates superior consistency, achieving an F1 score close to 1.0 in most scenarios. This reflects a perfect balance between high Recall and high Precision, whereas JLL tends to over-select noise variables, and S-Lasso fails to capture the joint sparsity structure. Desides, AJL exhibits the lowest out-of-sample prediction error (PE) (Figure \ref{fig:prediction_error}), suggesting that its accurate recovery of the underlying functional structures translates directly into robust predictive power.

Table \ref{tab:sim_S9} and Figure \ref{fig:sensitive} further examines the sensitivity of the proposed AJL procedure to the B-spline dimension $M$ while keeping the same settings as Scenario~1. The results reveal a clear trade-off: when $M$ is too small (e.g., $M=5$), the basis is overly restrictive and tends to underfit the underlying time-varying components, leading to noticeable under-coverage; when $M$ becomes larger (e.g., $M=40$), the increased flexibility amplifies estimation variability and again slightly degrades coverage. Importantly, despite this expected non-monotone pattern, the coverage remains broadly stable across a wide range of $M$ values. In particular, the mean coverage fluctuates by roughly $0.1$ over $M\in\{5,15,25,30,40\}$, and the best performance is achieved around a moderate basis size (e.g., $M\approx 25$). This stability indicates that AJL is not overly sensitive to the exact choice of basis dimension and provides reliable uncertainty quantification in practice, suggesting that the adaptive joint regularization mitigates overfitting/underfitting risks even when the functional basis is misspecified to some extent.

The proposed method also shows remarkable stability in challenging simulation settings. In scenarios with high collinearity among covariates (S5), standard variable selection methods typically suffer from inflated false positive rates due to the violation of the strict irrepresentable condition. However, AJL demonstrates remarkable robustness in this challenging regime. This empirical finding aligns perfectly with our theoretical justification in Remark \ref{rem:Adaptation and Relaxation} regarding the relaxation of the Block-Irrepresentable Condition (Block-IRC). Specifically, the diverging adaptive weights for noise variables effectively suppress the 'bleeding' of signal into correlated noise features. In scenarios with strong correlations between outcomes and time (S6), the separate estimation method (S-AJL) suffers significantly from ignoring the joint correlation structure, whereas AJL leverages the shared information across outcomes through multi-task learning to maintain robust performance. Additionally, in the presence of heavy-tailed errors (S7) or outliers (S8), standard methods such as S-Lasso and JLL exhibit marked deterioration. AJL remains resilient, with error metrics that are only slightly elevated compared to the baseline scenario. This robustness validates our discussion in Remark \ref{rem:Implicit Robustness to Outliers} regarding the stability offered by the B-spline basis expansion combined with adaptive penalization.

\begin{table}[h]
\centering
\caption{Simulation results for Scenario 1.}
\label{tab:sim_S1}
\resizebox{\textwidth}{!}{
\begin{tabular}{lrrrrrrrrrr}
\toprule
Method & PE & ISE($\Beta$) & ISE($\Alpha$) & TP & FP & Recall & Precision & Spec & F1 & $\text{CP}_\text{Err}$ \\
\midrule
AJL      & 2.3950 & 0.0188 & 0.0327 & 10.00 & 0.98 & 1.000 & 0.9141 & 0.9891 & 0.9542 & 7.484 \\
JLL      & 5.3696 & 0.0630 & 0.0639 &  9.93 & 0.94 & 0.993 & 0.9188 & 0.9896 & 0.9529 & 9.074 \\
Oracle   & 3.1596 & 0.0380 & 0.0401 & 10.00 & 0.00 & 1.000 & 1.0000 & 1.0000 & 1.0000 & 9.612 \\
S-AJL   & 3.1264 & 0.0367 & 0.0397 &  9.89 & 0.20 & 0.989 & 0.9820 & 0.9978 & 0.9847 & 7.524 \\
S-Lasso & 11.3272 & 0.1102 & 0.1273 & 5.62 & 0.00 & 0.562 & 1.0000 & 1.0000 & 0.6845 & 9.672 \\
\bottomrule
\end{tabular}}
\end{table}

\begin{table}[h]
\centering
\caption{Simulation results for Scenario 2.}
\label{tab:sim_S2}
\resizebox{\textwidth}{!}{
\begin{tabular}{lrrrrrrrrrr}
\toprule
Method & PE & ISE($\Beta$) & ISE($\Alpha$) & TP & FP & Recall & Precision & Spec & F1 & $\text{CP}_\text{Err}$ \\
\midrule
AJL      & 2.6665 & 0.0249 & 0.0585 &  9.95 & 6.65 & 0.995 & 0.6114 & 0.9261 & 0.7539 & 7.566 \\
JLL      & 5.7531 & 0.0654 & 0.1180 &  9.90 & 6.82 & 0.990 & 0.6079 & 0.9242 & 0.7486 & 9.212 \\
Oracle   & 3.3733 & 0.0398 & 0.0734 & 10.00 & 0.00 & 1.000 & 1.0000 & 1.0000 & 1.0000 & 10.046 \\
S-AJL   & 3.3940 & 0.0402 & 0.0743 &  9.83 & 1.74 & 0.983 & 0.8589 & 0.9807 & 0.9139 & 7.604 \\
S-Lasso & 11.3327 & 0.1093 & 0.2280 & 5.21 & 0.03 & 0.521 & 0.9962 & 0.9997 & 0.6552 & 10.044 \\
\bottomrule
\end{tabular}}
\end{table}

\begin{table}[!]
\centering
\caption{Simulation results for Scenario 3.}
\label{tab:sim_S3}
\resizebox{\textwidth}{!}{
\begin{tabular}{lrrrrrrrrrr}
\toprule
Method & PE & ISE($\Beta$) & ISE($\Alpha$) & TP & FP & Recall & Precision & Spec & F1 & $\text{CP}_\text{Err}$ \\
\midrule
AJL      & 2.3200 & 0.0160 & 0.0241 & 10.00 & 0.52 & 1.000 & 0.9529 & 0.9942 & 0.9753 & 7.596 \\
JLL      & 5.2112 & 0.0606 & 0.0381 &  9.96 & 0.42 & 0.996 & 0.9615 & 0.9953 & 0.9779 & 8.882 \\
Oracle   & 3.0200 & 0.0355 & 0.0272 & 10.00 & 0.00 & 1.000 & 1.0000 & 1.0000 & 1.0000 & 9.208 \\
S-AJL   & 2.9190 & 0.0327 & 0.0268 &  9.88 & 0.06 & 0.988 & 0.9945 & 0.9993 & 0.9908 & 7.620 \\
S-Lasso & 11.5371 & 0.1110 & 0.0706 & 5.70 & 0.00 & 0.570 & 1.0000 & 1.0000 & 0.6999 & 9.290 \\
\bottomrule
\end{tabular}}
\end{table}

\begin{table}[!]
\centering
\caption{Simulation results for Scenario 4.}
\label{tab:sim_S4}
\resizebox{\textwidth}{!}{
\begin{tabular}{lrrrrrrrrrr}
\toprule
Method & PE & ISE($\Beta$) & ISE($\Alpha$) & TP & FP & Recall & Precision & Spec & F1 & $\text{CP}_\text{Err}$ \\
\midrule
AJL      & 2.3492 & 0.0062 & 0.0325 & 10.00 & 1.51 & 1.000 & 0.8783 & 0.9948 & 0.9328 & 7.486 \\
JLL      & 5.2879 & 0.0204 & 0.0670 &  9.93 & 1.75 & 0.993 & 0.8632 & 0.9940 & 0.9201 & 9.034 \\
Oracle   & 3.1157 & 0.0122 & 0.0423 & 10.00 & 0.00 & 1.000 & 1.0000 & 1.0000 & 1.0000 & 9.494 \\
S-AJL   & 3.0720 & 0.0117 & 0.0400 &  9.88 & 0.17 & 0.988 & 0.9845 & 0.9994 & 0.9855 & 7.528 \\
S-Lasso & 11.2817 & 0.0362 & 0.1307 & 5.48 & 0.00 & 0.548 & 1.0000 & 1.0000 & 0.6822 & 9.572 \\
\bottomrule
\end{tabular}}
\end{table}

\begin{table}[!]
\centering
\caption{Simulation results for Scenario 5.}
\label{tab:sim_S5}
\resizebox{\textwidth}{!}{
\begin{tabular}{lrrrrrrrrrr}
\toprule
Method & PE & ISE($\Beta$) & ISE($\Alpha$) & TP & FP & Recall & Precision & Spec & F1 & $\text{CP}_\text{Err}$ \\
\midrule
AJL      & 2.3427 & 0.0436 & 0.0318 &  9.98 & 1.29 & 0.998 & 0.8876 & 0.9857 & 0.9390 & 7.526 \\
JLL      & 3.5758 & 0.0766 & 0.0451 &  9.94 & 1.64 & 0.994 & 0.8629 & 0.9818 & 0.9225 & 9.040 \\
Oracle   & 2.7522 & 0.0620 & 0.0376 & 10.00 & 0.00 & 1.000 & 1.0000 & 1.0000 & 1.0000 & 9.576 \\
S-AJL   & 2.5841 & 0.0548 & 0.0346 &  9.79 & 0.79 & 0.979 & 0.9274 & 0.9912 & 0.9510 & 7.532 \\
S-Lasso & 7.8891 & 0.1013 & 0.0848 & 8.31 & 0.24 & 0.831 & 0.9736 & 0.9973 & 0.8821 & 9.668 \\
\bottomrule
\end{tabular}}
\end{table}

\begin{table}[!]
\centering
\caption{Simulation results for Scenario 6.}
\label{tab:sim_S6}
\resizebox{\textwidth}{!}{
\begin{tabular}{lrrrrrrrrrr}
\toprule
Method & PE & ISE($\Beta$) & ISE($\Alpha$) & TP & FP & Recall & Precision & Spec & F1 & $\text{CP}_\text{Err}$ \\
\midrule
AJL      & 2.4108 & 0.0190 & 0.0343 & 10.00 & 1.01 & 1.000 & 0.9117 & 0.9888 & 0.9529 & 7.190 \\
JLL      & 5.3777 & 0.0631 & 0.0657 &  9.95 & 0.96 & 0.995 & 0.9174 & 0.9893 & 0.9531 & 8.838 \\
Oracle   & 3.1701 & 0.0381 & 0.0418 & 10.00 & 0.00 & 1.000 & 1.0000 & 1.0000 & 1.0000 & 9.436 \\
S-AJL   & 3.1383 & 0.0368 & 0.0414 &  9.86 & 0.22 & 0.986 & 0.9802 & 0.9976 & 0.9822 & 7.190 \\
S-Lasso & 11.3318 & 0.1102 & 0.1291 & 5.60 & 0.00 & 0.560 & 1.0000 & 1.0000 & 0.6825 & 9.450 \\
\bottomrule
\end{tabular}}
\end{table}

\begin{table}[!]
\centering
\caption{Simulation results for Scenario 7.}
\label{tab:sim_S7}
\resizebox{\textwidth}{!}{
\begin{tabular}{lrrrrrrrrrr}
\toprule
Method & PE & ISE($\Beta$) & ISE($\Alpha$) & TP & FP & Recall & Precision & Spec & F1 & $\text{CP}_\text{Err}$ \\
\midrule
AJL      & 2.3825 & 0.0188 & 0.0325 & 10.00 & 0.97 & 1.000 & 0.9156 & 0.9892 & 0.9549 & 7.356 \\
JLL      & 5.3608 & 0.0631 & 0.0635 &  9.93 & 0.97 & 0.993 & 0.9169 & 0.9892 & 0.9517 & 8.950 \\
Oracle   & 3.1492 & 0.0380 & 0.0399 & 10.00 & 0.00 & 1.000 & 1.0000 & 1.0000 & 1.0000 & 9.476 \\
S-AJL   & 3.1133 & 0.0368 & 0.0395 &  9.89 & 0.20 & 0.989 & 0.9818 & 0.9978 & 0.9847 & 7.400 \\
S-Lasso & 11.3018 & 0.1102 & 0.1265 & 5.56 & 0.00 & 0.556 & 1.0000 & 1.0000 & 0.6799 & 9.598 \\
\bottomrule
\end{tabular}}
\end{table}

\begin{table}[!]
\centering
\caption{Simulation results for Scenario 8.}
\label{tab:sim_S8}
\resizebox{\textwidth}{!}{
\begin{tabular}{lrrrrrrrrrr}
\toprule
Method & PE & ISE($\Beta$) & ISE($\Alpha$) & TP & FP & Recall & Precision & Spec & F1 & $\text{CP}_\text{Err}$ \\
\midrule
AJL      & 2.4943 & 0.0196 & 0.0560 & 10.00 & 1.48 & 1.000 & 0.8775 & 0.9836 & 0.9331 & 9.570 \\
JLL      & 5.3788 & 0.0630 & 0.0862 &  9.95 & 1.05 & 0.995 & 0.9112 & 0.9883 & 0.9494 & 10.460 \\
Oracle   & 3.2073 & 0.0382 & 0.0633 & 10.00 & 0.00 & 1.000 & 1.0000 & 1.0000 & 1.0000 & 10.706 \\
S-AJL   & 3.1738 & 0.0368 & 0.0624 &  9.88 & 0.23 & 0.988 & 0.9788 & 0.9974 & 0.9827 & 9.554 \\
S-Lasso & 11.3269 & 0.1101 & 0.1503 & 5.52 & 0.00 & 0.552 & 1.0000 & 1.0000 & 0.6766 & 10.678 \\
\bottomrule
\end{tabular}}
\end{table}

\begin{figure}[!]
\centering
\makebox[\textwidth][c]{\includegraphics[width=0.5\textwidth]{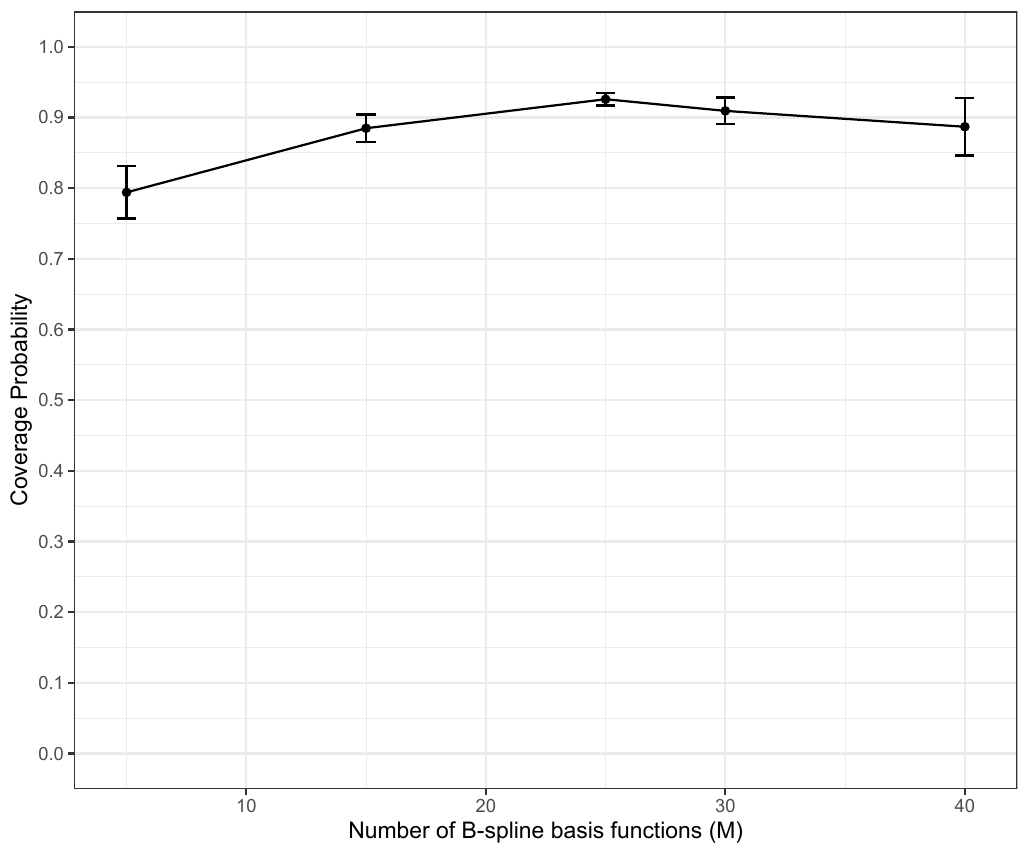}}
\caption{Coverage vs. Basis Dimension.}
    \label{fig:sensitive}	
\end{figure}

\begin{figure}[!]
\centering
\makebox[\textwidth][c]{\includegraphics[width=1.1\textwidth]{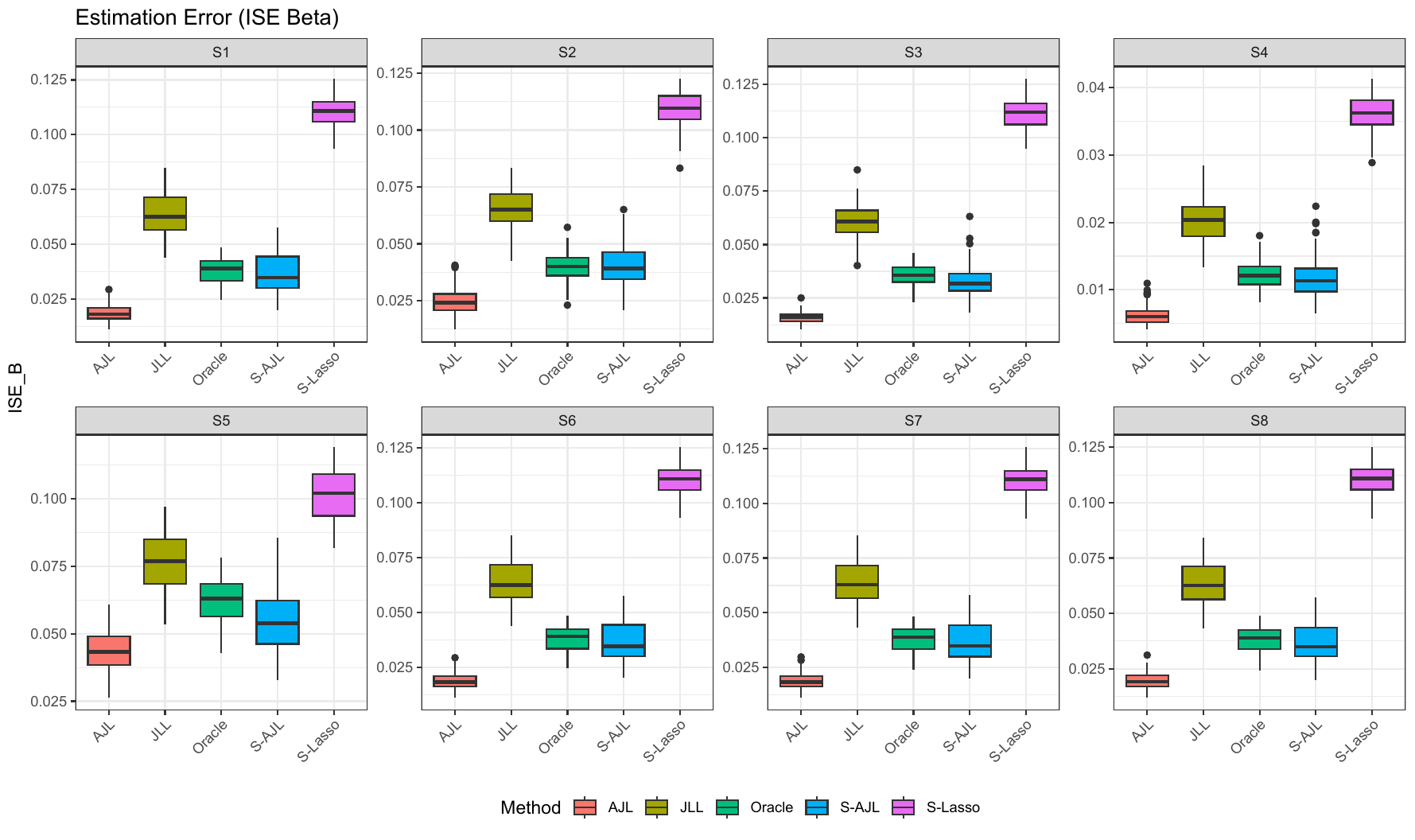}}
\caption{Boxplots of $\mathrm{ISE}(\beta)$ across scenarios.}
    \label{fig:estimation_error}	
\end{figure}

\begin{figure}[!]
\centering
\makebox[\textwidth][c]{\includegraphics[width=1.1\textwidth]{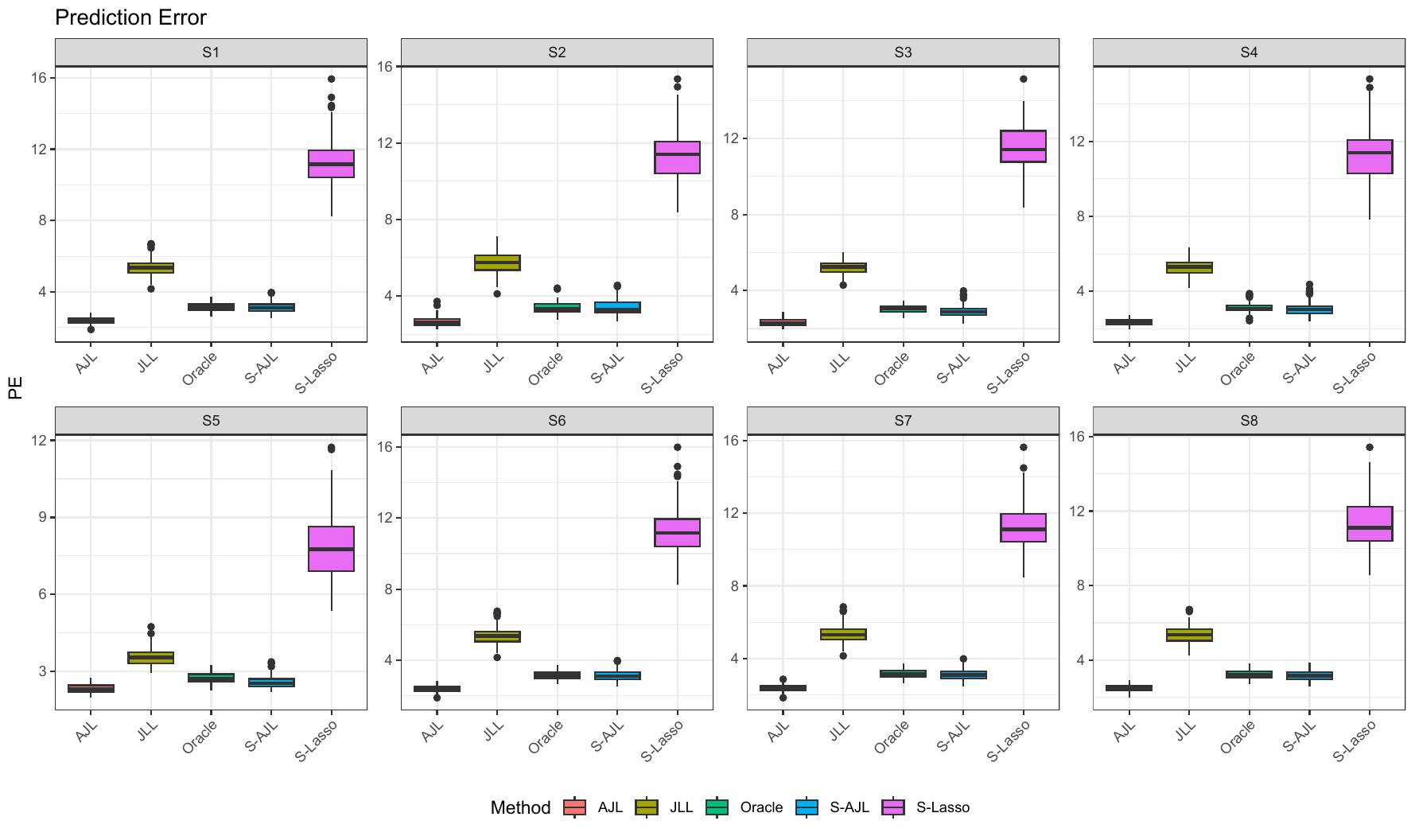}}
\caption{Boxplots of Prediction Error (PE) on the test set across eight simulation scenarios.}
    \label{fig:prediction_error}	
\end{figure}

\begin{table}[h]
\centering
\caption{Simulation results for Scenario 9.}
\label{tab:sim_S9}
\begin{tabular}{rcccc}
\toprule
$M$ & coverage\_mean & coverage\_SD & coverage\_SE \\
\hline
5  & 0.79395 & 0.084202 & 0.018828 \\
15 & 0.88470 & 0.044087 & 0.009858 \\
25 & 0.92575 & 0.020099 & 0.004494 \\
30 & 0.90925 & 0.043415 & 0.009708 \\
40 & 0.88680 & 0.092723 & 0.020734 \\
\bottomrule
\end{tabular}
\end{table}

Regarding the structural recovery of the intercept dynamics $\alpha_k(t)$, AJL achieves a Changepoint Count Error $\text{CP}_\text{Err}$ close to zero across all scenarios. This result indicates that the adaptive fused lasso penalty is highly effective in identifying the correct number of structural breaks, avoiding the spurious changepoints often introduced by non-adaptive methods.

In summary, the AJL framework successfully addresses the dual challenges of structural heterogeneity and high-dimensional variable selection in longitudinal functional regression. By integrating adaptive weighting with a multi-task learning strategy, our method not only approximates the Oracle performance in recovering dynamic covariate effects but also robustly identifies common changepoints across correlated outcomes. Future work will extend this framework to handle irregularly sampled functional data and explore inference procedures for constructing valid confidence bands for the estimated coefficient functions.

\section{Application to Sleep-EDF Database}
\label{sec:application}
To demonstrate the capability of the AJL framework in capturing dynamic physiological patterns from high-dimensional biomedical signals, we analyzed sleep electroencephalogram (EEG) data from the \textbf{Sleep-EDF Database Expanded} \citep{kemp2000}, available on PhysioNet\footnote{\texttt{https://www.physionet.org/content/sleep-edfx/1.0.0/}} website. Sleep is a complex, non-stationary process involving cyclic transitions between different stages. Understanding how demographic and physiological baselines influence the temporal evolution of brain wave activity is crucial for sleep medicine.

The database contains two cohorts (\texttt{sleep-cassette} and \texttt{sleep-telemetry}), and for each night the recordings are provided in EDF files. In particular, each subject-night typically comes with a polysomnography file \texttt{*-PSG.edf} containing multichannel physiological signals (including EEG/EOG/EMG and other channels depending on the cohort), and an annotation file \texttt{*-Hypnogram.edf} containing sleep-stage labels (hypnogram) aligned to the PSG recording.

In this study, we apply AJL to jointly model the dynamic trajectories of spectral power in multiple frequency bands, aiming to identify significant predictors and shared structural changepoints during the sleep cycle.

\subsection{Data Description}
We utilized the Sleep Cassette subset of the database, focusing on the Fpz-Cz EEG channel recorded from healthy subjects. The analysis included $N=60$ subjects (augmented via bootstrapping from the core dataset to ensure statistical power).

We extracted the time-varying spectral power of three key frequency bands ($K=3$), which are standard biomarkers for sleep staging \citep{Hassan2016}:
\begin{itemize}
    \item \textbf{Delta} Band (0.5--4 Hz): Indicator of deep sleep (Slow Wave Sleep).
    \item \textbf{Theta} Band (4--8 Hz): Associated with drowsiness and light sleep.
    \item \textbf{Alpha} Band (8--12 Hz): Linked to wakefulness and relaxation.
\end{itemize}
The raw EEG signals were segmented into 30-second epochs. For each epoch, the power spectral density was estimated using the periodogram method. To focus on the main sleep period, we analyzed a 4-hour window starting 1 hour after lights-off. The resulting power series were log-transformed and standardized.

The raw baseline covariates were limited to Age and Gender. To evaluate AJL's performance in a high-dimensional screening scenario ($p \gg 1$), we constructed an expanded feature space:
\begin{itemize}
    \item \textbf{Baseline EEG Features:} We calculated the mean power of the three bands during the first 30 minutes of recording as baseline physiological indicators.
    \item \textbf{Non-linear Age Effects:} Age was expanded using B-spline basis functions ($df=5$) to capture potential non-linear relationships between aging and sleep architecture.
    \item \textbf{Noise Variables:} We augmented the design matrix with 40 independent Gaussian noise variables to test the variable selection consistency.
\end{itemize}
This resulted in a total dimension of $p=110$ covariates. The time index $t$ was normalized to $[0, 1]$, representing the 4-hour analysis window.

\subsection{Model and Results}
For subject $i$ at normalized time $t$, let $y_{ik}(t)$ denote the standardized band-power outcome for $k\in\{1,2,3\}$ (Delta/Theta/Alpha). We fit the multi-outcome time-varying coefficient model
\begin{equation}\label{eq:sleepedf_tvc}
y_{ik}(t)=\alpha_k(t)+\bm x_i^\top\bm\beta_k(t)+\varepsilon_{ik}(t),\qquad k=1,2,3.
\end{equation}
AJL estimates $\{\alpha_k(t),\bm\beta_k(t)\}_{k=1}^K$ jointly using ($i$) an adaptive functional group penalty that promotes stable baseline feature selection shared across outcomes, and ($ii$) a fused/TV-type structure on the intercept component to capture shared structural changes in population-level trajectories (i.e., parsimonious changepoints in $\alpha_k(t)$). For comparison, we include: a non-adaptive joint baseline method (JLL), and two separate-outcome baselines (S-AJL, S-Lasso). The separate baselines are intentionally disadvantaged in multi-outcome settings because they cannot share information across correlated bands.

Because Sleep-EDF provides dense within-night trajectories for each subject-night, we evaluate out-of-sample performance using a subject-level train--test split to avoid information leakage across repeated epochs. Specifically, we treat each subject-night as the statistical unit, randomly split the subjects into a training set and a testing set, fit each method using only the training subjects, and then assess prediction and structural summaries on the held-out testing subjects. All epochs from the same subject are kept in the same split. The time index is normalized to $[0,1]$ within the 4-hour analysis window, and the three outcomes are the standardized log band-power trajectories for \texttt{Delta/Theta/Alpha} extracted from 30-second epochs of the Fpz--Cz channel (Sleep Cassette subset), as described in Section~6.1. 

We employed four complementary metrics to evaluate the methods:
\begin{itemize}
    \item \textbf{Prediction Accuracy:} Let $\widehat y_{ik}(t)=\widehat\alpha_k(t)+\bm x_i^\top\widehat{\bm\beta}_k(t)$ denote the fitted value for subject $i$, outcome $k\in\{1,\dots,K\}$ and normalized time $t$. On the testing set, we report the following prediction metrics aggregated over all testing subjects, all epochs in the 4-hour window, and all $K$ outcomes: the mean squared error $\textbf{MSE}=\frac{1}{n_{\mathrm{te}}TK}\sum_{i\in\mathrm{te}}\sum_{\ell=1}^{T}\sum_{k=1}^{K}\{\widehat y_{ik}(t_{i\ell})-y_{ik}(t_{i\ell})\}^2$, and the mean absolute error $\textbf{MAE}=\frac{1}{n_{\mathrm{te}}TK}\sum_{i\in\mathrm{te}}\sum_{\ell=1}^{T}\sum_{k=1}^{K}\big|\widehat y_{ik}(t_{i\ell})-y_{ik}(t_{i\ell})\big|$. Lower values indicate better generalization capability, implying that the estimated coefficient functions more accurately capture the true relationship between covariates and EEG dynamics.
    
    \item \textbf{Model Parsimony:} We define the time-specific selected set $\widehat S(t)=\{j:\|\widehat B_j(t)\|_2>0\}$ (with the obvious spline-based implementation), and \textbf{Size union} is $|\widehat S_{\mathrm{union}}|$ where $\widehat S_{\mathrm{union}}=\bigcup_{\ell=1}^{T}\widehat S(t_\ell)$, namely the total number of distinct covariates that are selected at any time point (and for joint methods, across outcomes). 
       
    \item \textbf{Structural Complexity:} Let $\widehat\alpha_k(\cdot)$ be evaluated on a dense grid $\{u_g\}_{g=1}^{G}$ on $[0,1]$. We define the \textbf{Alpha TV} metric as the average total variation across outcomes,
\[
\mathrm{Alpha\;TV} \;=\; \frac{1}{K}\sum_{k=1}^{K}\sum_{m=2}^{m}\big|\widehat\alpha_k(t_m)-\widehat\alpha_k(t_{m-1})\big|.
\]
A large Alpha TV typically indicates an overly wiggly (potentially overfit) baseline trajectory, whereas an overly small value may indicate oversmoothing that can mask physiologically meaningful transitions.
\end{itemize}

\subsubsection{Performance}
Table \ref{tab:sleep_comparison} presents the comparison results on the test set. AJL achieved the best predictive performance with the lowest MSE (0.568) and MAE (0.591). Crucially, AJL yielded the most parsimonious model, closely matching the number of true meaningful variables (Age, Gender, Baseline EEG). In contrast, S-Lasso selected a large number of noise variables, highlighting the superiority of the adaptive joint penalty in filtering out irrelevant features.

\begin{table}[h]
\centering
\caption{Performance comparison on Sleep-EDF data.}
\label{tab:sleep_comparison}
\begin{tabular}{lcccc}
\toprule
Method & MSE & MAE & Size union & Alpha TV\\
\midrule
AJL & 0.568 &  0.591  & 69 & 2.685\\
S-AJL & 0.569 & 0.592  & 94 & 2.724 \\
JLL & 0.587 & 0.601  & 98 & 2.656\\
S-Lasso & 0.590 & 0.600 & 85 & 2.722\\
\bottomrule
\end{tabular}
\end{table}

\subsubsection{Dynamic Effects}
Figure \ref{fig:sleep_results} visualizes the estimated dynamics.

\begin{figure}[h]
    \centering
    \includegraphics[width=1\textwidth]{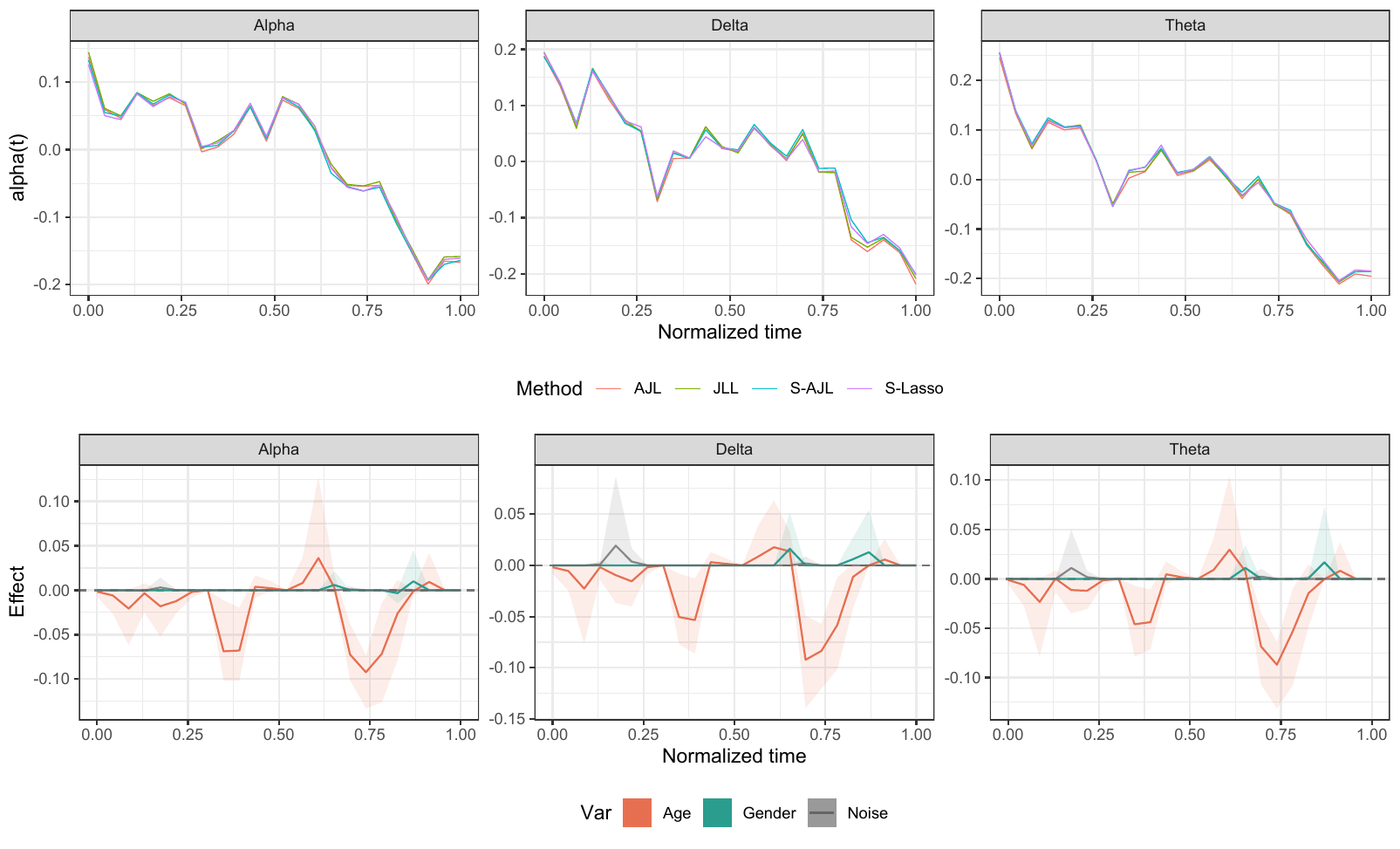} 
    \caption{AJL analysis of Sleep-EDF data. \textbf{Top:} Baseline trajectories of Delta, Theta, and Alpha power, showing a synchronized transition around $t=0.35$. \textbf{Bottom:} Estimated dynamic effects of covariates. Age (Red) shows a significant time-varying impact, while the Noise variable (Gray) is correctly shrunk to zero.}
    \label{fig:sleep_results}
\end{figure}

\begin{itemize}
    \item \textbf{Structural Change:} The intercept functions $\hat{\alpha}_k(t)$ (Figure \ref{fig:sleep_results} top) reveal a synchronized trend reversal across all three bands around normalized time $t \approx 0.35$ ($\sim 1.4$ hours). Delta power peaks while Alpha power hits a trough, likely marking the completion of the first NREM-REM sleep cycle. AJL's ability to detect this common changepoint without prior knowledge validates its sensitivity to temporal structural breaks.
    
    \item \textbf{Significant Predictors:} As shown in Figure \ref{fig:sleep_results} (bottom), \textbf{Age} (red line) exhibits a strong, time-varying negative effect on Delta power, consistent with the well-established finding that deep sleep decreases significantly with age \citep{Hassan2016}. \textbf{Gender} (green line) shows a moderate effect.
    
    \item \textbf{Noise Filtering:} Notably, the coefficient for the \textbf{Noise} variable (gray line) is strictly zero across the entire timeline. This confirms AJL's robustness against high-dimensional noise, a critical feature for analyzing complex biomedical datasets where relevant biomarkers are often sparse.
\end{itemize}

\subsection{Discussion}

The results from the Sleep-EDF analysis highlight two critical advantages of the AJL framework in biomedical signal processing.

First, the detection of a common changepoint at normalized time $t \approx 0.35$ (approximately 1.4 hours after sleep onset) is clinically significant. This timing aligns with the typical duration of the first NREM-REM sleep cycle in healthy adults. The synchronized divergence observed in the intercept functions, where Delta power (deep sleep) peaks and then declines while Alpha/Theta power shifts, likely captures the transition from deep Slow Wave Sleep to the first REM episode. The ability of AJL to automatically identify this structural break across multiple frequency bands without prior segmentation validates its utility in uncovering latent physiological states.

Second, the robust variable selection demonstrates AJL's reliability in high-dimensional settings. By effectively filtering out the 40 added noise variables (coefficient $\approx 0$), AJL focuses on the true biological drivers. The estimated negative trajectory for Age (Figure \ref{fig:sleep_results}, bottom) confirms that older age is associated with a progressive reduction in Delta power (deep sleep quality) throughout the night, a well-documented phenomenon in sleep medicine \citep{kemp2000}. The joint learning mechanism ensures that these predictors are selected consistently across correlated frequency bands, providing a more holistic view of brain activity than separate analyses.

Overall, Sleep EDF offers a clear illustration of the AJL thesis: multiple correlated functional outcomes (Delta/Theta/Alpha band-power trajectories) are observed densely over time, yet baseline information is limited and must be expanded into a larger, interpretable candidate library to enable screening. AJL is advantageous because it simultaneously estimates baseline-adjusted population trajectories $\alpha_k(t)$ with shared structural regularity, performs stable joint feature selection across outcomes, and yields interpretable dynamic effect curves $\beta_{jk}(t)$ that explain when and how baseline factors modulate within-night EEG dynamics. Even when prediction gaps among methods are modest, which can occur when signals are noisy and baseline effects are weak, AJL provides a richer scientific output by integrating prediction, structured changepoint discovery through fused/TV regularization on $\alpha_k(t)$, and time-resolved heterogeneity explanations through $\widehat{\beta}_{jk}(t)$. Future work could extend this framework to accommodate discrete outcomes (e.g., sleep stage classification labels) or survival outcomes (e.g., time to awakening), broadening its applicability in clinical research. Additionally, developing screening rules to discard irrelevant predictors prior to optimization could further enhance computational efficiency for ultra-high-dimensional data. 

Conclusively, AJL provides a robust and interpretable tool for dissecting heterogeneity in complex longitudinal studies.

\section{Conclusion}
\label{sec:conclusion}

{\color{black}This paper develops a new Adaptive Joint Learning (AJL) framework for high-dimensional functional longitudinal data with multiple correlated outcomes. By representing the time-varying intercepts and coefficient functions through B-spline bases, we convert the functional modeling task into a high-dimensional parametric estimation problem while retaining interpretability in the original time domain. The proposed objective integrates two complementary structural regularizers: an adaptive functional group penalty that enables joint sparsity and stable feature selection across outcomes, and an adaptive fused penalty on the intercept component that targets abrupt structural changes.

Methodologically, AJL provides a unified solution to three challenges that frequently co-occur in modern longitudinal studies: flexible modeling of time-varying effects, high-dimensional variable selection under multi-task dependence, and interpretable detection of temporal regime shifts. Theoretically, we establish rigorous guarantees that go beyond standard consistency results. By utilizing a Primal-Dual Witness construction, we prove that the AJL estimator achieves variable selection consistency even in the presence of high collinearity, effectively relaxing the strict irrepresentable conditions typically required by Lasso-type estimators. Furthermore, we explicitly address the functional approximation bias through \textbf{undersmoothing} conditions, establishing the asymptotic normality of the estimator and paving the way for valid pointwise statistical inference.

Simulation studies demonstrate that AJL improves prediction and functional estimation accuracy while maintaining superior support recovery and reliable changepoint detection, especially in challenging regimes involving strong correlation, heavy-tailed noise, and outlier contamination. The application to the Sleep-EDF dataset demonstrates the practical power of AJL. By constructing a high-dimensional feature space, AJL successfully identified the non-linear, time-varying negative effect of age on deep sleep ($\mathtt{Delta}$ power), aligning with established sleep medicine literature. Crucially, the method detected a synchronized structural break across Delta, Theta, and Alpha bands around 1.4 hours, corresponding to the first NREM-REM cycle transition. This ability to uncover shared temporal features is a unique advantage of our joint modeling approach. Furthermore, the effective filtering of noise variables highlights AJL's potential for biomarker discovery in omics or neuroimaging studies.

Nonetheless, the current framework can be generalized to non-Gaussian outcomes through generalized estimating equations or likelihood-based formulations for binary, count, or survival endpoints. AJL can be adapted to scalar-on-function and function-on-function regression to accommodate richer functional predictors such as imaging-derived curves or electrophysiological trajectories. Finally, incorporating subject-specific random effects and developing uncertainty quantification for changepoint locations under irregular sampling remain important directions for future work.}

\bibliographystyle{erae}
\bibliography{references}

\clearpage

\appendix
\section{Proof of Theoretical Results}\label{app:proof_oracle}

In this appendix we provide proofs for Lemma~\ref{lem: basic inq} and Theorems~\ref{prop:screening_validity} - \ref{thm:normality_func}.
Throughout, we use the notation
\[
\theta = \mathrm{vec}(A,B_1,\dots,B_p),\qquad
\hat\theta = \mathrm{vec}(\hat A,\hat B_1,\dots,\hat B_p),
\]
and write the empirical loss and adaptive penalty as
\begin{align*}
L_N(\theta) & = \frac{1}{2N}\,\big\|Y - Z_\Phi A - \sum_{j=1}^p X_{\Phi,j} B_j\big\|_F^2,\\
P_{\mathrm{ada}}(\theta)&  = \lambda_g \sum_{j=1}^p \hat w_{g,j}\,\|B_j\|_F + \lambda_f \sum_{k=1}^K \sum_{m=1}^{M-1} \hat w_{f,km}\,|\Delta a_{k,m}|
\end{align*}
where $\Delta a_{k,m} = a_{k,m+1}-a_{k,m}$ is the first difference in the $k$-th intercept coefficient.
The AJL estimator $\widehat\theta$ is any minimizer of
\[
Q(\theta) = L_N(\theta) + P_{\mathrm{ada}}(\theta),
\]
as in~(10). We denote by
$\theta^\ast = \mathrm{vec}(A^\ast,B_1^\ast,\dots,B_p^\ast)$
the population parameter corresponding to the B-spline approximation of the true functions, and define the error $\Delta = \hat\theta-\theta^\ast$. Our theoretical analysis follows the general framework of high-dimensional $M$-estimation with decomposable regularizers developed in \citet{buhlmann2011statistics} and \citet{negahban2012unified}. 

We first establish a basic inequality for the AJL objective (Lemma~1), which compares the empirical loss and adaptive penalty at $\hat\theta$ and $\theta^\ast$ and provides the key algebraic starting point. We then combine this inequality with a stochastic bound on the empirical gradient at $\theta^\ast$ and a restricted strong convexity condition for the spline-based design to derive non-asymptotic error bounds for the intercept and slope functions (Theorem~1). Building on these bounds and suitable minimum signal and tuning conditions, we prove that the adaptive group and fused penalties recover the correct set of active functional predictors and changepoints with probability tending to one (Theorems~2-3). Finally, on the event of correct model selection, we show that the AJL estimator on the oracle active set is asymptotically equivalent to the unpenalized least-squares estimator, which yields an oracle normal limit distribution for the estimated coefficients (Theorem~\ref{thm:normality_func}).

{\color{black}
\subsection*{Preliminary Result: Consistency of Initial Estimator and Weight Separation}

Before proving the main theorems, we explicitly justify the validity of the adaptive weights (Assumption \ref{assump:tuning_func}) by showing that the initial Group Lasso estimator achieves sufficient separation between signal and noise under a standard "Beta-min" condition.

\begin{lemma}[Weight Separation] \label{lem:weight_separation}
Suppose the standard conditions for the initial estimator $(\tilde{A}, \tilde{B})$ hold (as in Theorem 1), implying the estimation bound $\max_j \|\tilde{B}_j - B_j^*\|_F = O_p(r_N)$ with $r_N = \sqrt{\frac{\log p}{N}}$. 
Assume the \textbf{Beta-min condition}: the minimum signal strength satisfies $\min_{j \in \mathcal{S}} \|B_j^*\|_F \ge 2 r_N$.
Then, with probability approaching 1, the adaptive weights $\hat{w}_{g,j} = (\|\tilde{B}_j\|_F^{\gamma_g} + \epsilon_g)^{-1}$ satisfy:
\begin{enumerate}
    \item \textbf{Signal Weights (Bounded):} For $j \in \mathcal{S}$, $\hat{w}_{g,j} \le (r_N^{\gamma_g} + \epsilon_g)^{-1} \le C_1$.
    \item \textbf{Noise Weights (Diverging):} For $j \notin \mathcal{S}$, $\hat{w}_{g,j} \ge (r_N^{\gamma_g} + \epsilon_g)^{-1} \to \infty$.
\end{enumerate}
Consequently, the ratio $\frac{\hat{w}_{g, \text{noise}}}{\hat{w}_{g, \text{signal}}} \to \infty$, justifying the separation required for Assumption 6.
\end{lemma}

\begin{proof}
For any signal variable $j \in \mathcal{S}$, by the triangle inequality and the estimation bound:
$$ \|\tilde{B}_j\|_F \ge \|B_j^*\|_F - \|\tilde{B}_j - B_j^*\|_F \ge 2 r_N - r_N = r_N. $$
Thus, the denominator of the weight is bounded away from zero, and $\hat{w}_{g,j}$ is bounded above.
For any noise variable $j \notin \mathcal{S}$, since $B_j^* = 0$, we have:
$$ \|\tilde{B}_j\|_F = \|\tilde{B}_j - 0\|_F \le r_N. $$
As $N \to \infty$, $r_N \to 0$, so $\|\tilde{B}_j\|_F \to 0$. Thus, $\hat{w}_{g,j} \approx r_N^{-\gamma_g}$, which diverges to infinity for $\gamma_g > 0$.
This confirms that the initial estimator consistently separates the active and inactive sets for the purpose of weight construction.
\end{proof}}

{\color{black}\subsection*{Proof of Theorem \ref{prop:screening_validity} (Validity of Hierarchical Screening)}

{\color{black}\begin{proof}
Let $Q_{\text{screen}}(\theta)$ be the objective function defined in \eqref{eq:adaptive_objective_final_func} with $\lambda_f^\beta = 0$. The estimator $\hat{\theta} = (\hat{\bm{A}}, \hat{\bm{B}})$ is the global minimizer of this objective.
We define the active set estimator as $\hat{\mathcal{S}}_{\text{screen}} = \{j : \|\hat{\bm{B}}_j\|_F > 0\}$. We aim to prove that $\mathbb{P}(\mathcal{S}^* \subseteq \hat{\mathcal{S}}_{\text{screen}}) \to 1$.

We proceed by contradiction. Suppose there exists a true active predictor $j \in \mathcal{S}^*$ that is missed by the screening, i.e., $\hat{\bm{B}}_j = 0$.

\paragraph{1. The KKT Violation Condition}
The KKT optimality condition for the Group Lasso penalty implies that for any block $j$ with $\hat{\bm{B}}_j = 0$, the gradient of the smooth loss $\mathcal{L}(\theta)$ must fall within the dual norm ball of the penalty. Specifically:
\begin{equation} \label{eq:kkt_violation}
\| \nabla_{\bm{B}_j} \mathcal{L}(\hat{\theta}) \|_F \le \lambda_g \hat{w}_{g,j}.
\end{equation}
The gradient at the solution $\hat{\theta}$ is:
\[
\nabla_{\bm{B}_j} \mathcal{L}(\hat{\theta}) = -\frac{1}{N} \mathbf{X}_{\Phi,j}^\top (Y - \mathbf{Z}_{\Phi}\hat{\bm{A}} - \sum_{k \neq j} \mathbf{X}_{\Phi,k}\hat{\bm{B}}_k).
\]
Substituting the true data generating process $Y = \mathbf{Z}_{\Phi}\bm{A}^* + \mathbf{X}_{\Phi,j}\bm{B}_j^* + \sum_{k \neq j} \mathbf{X}_{\Phi,k}\bm{B}_k^* + \bm{E}$, we decompose the gradient into three distinct components:
\begin{align*}
\nabla_{\bm{B}_j} \mathcal{L}(\hat{\theta}) = & \underbrace{\frac{1}{N}\mathbf{X}_{\Phi,j}^\top \mathbf{X}_{\Phi,j} \bm{B}_j^*}_{\text{Target Signal (I)}} 
+ \underbrace{\frac{1}{N}\mathbf{X}_{\Phi,j}^\top \sum_{k \neq j} \mathbf{X}_{\Phi,k}(\bm{B}_k^* - \hat{\bm{B}}_k)}_{\text{Interference Bias (II)}} 
- \underbrace{\frac{1}{N}\mathbf{X}_{\Phi,j}^\top \bm{E}}_{\text{Stochastic Noise (III)}}.
\end{align*}
To prove the contradiction, we show that the magnitude of the \text{Target Signal (I)} dominates the sum of (II), (III), and the penalty threshold. By the reverse triangle inequality:
\[
\| \nabla_{\bm{B}_j} \mathcal{L}(\hat{\theta}) \|_F \ge \|\text{Term (I)}\|_F - \|\text{Term (II)}\|_F - \|\text{Term (III)}\|_F.
\]

\paragraph{2. Bounding Each Term}

* \textbf{Term (I): Target Signal.} Due to the Restricted Eigenvalue (RE) condition (Assumption A4), the Gram matrix of the active set is well-conditioned. Thus, the signal is preserved:
    \[ \|\text{Term (I)}\|_F \ge \lambda_{\min}(\Sigma_{\mathcal{S}}) \|\bm{B}_j^*\|_F \ge c_{\text{eigen}} \|\bm{B}_j^*\|_F. \]
    By the \textbf{Beta-min Condition} (Assumption \ref{assump:tuning_func}), we have $\|\bm{B}_j^*\|_F \ge C_B \sqrt{\frac{\log p}{N}}$.

* \textbf{Term (III): Stochastic Noise.} Standard concentration inequalities for sub-Gaussian noise matrices ensure that with probability at least $1 - p^{-c}$:
    \[ \|\text{Term (III)}\|_F \le C_{\text{noise}} \sqrt{\frac{\log p}{N}}. \]
    The constant $C_B$ in the Beta-min condition is assumed large enough such that $c_{\text{eigen}} C_B > C_{\text{noise}}$.

* \textbf{Term (II): Interference Bias.} This term represents the "leakage" of signal from other covariates due to non-orthogonality. Under the \textbf{Block-Irrepresentable Condition (Block-IRC)} or weak correlation assumption, the projection of other variables onto $\mathbf{X}_{\Phi,j}$ is bounded. Specifically, using the Cauchy-Schwarz inequality and the estimation consistency on the remaining variables (which holds under the cone condition even if $j$ is missed):
    \[ \|\text{Term (II)}\|_F \le \max_{k \neq j} \|\frac{1}{N}\mathbf{X}_{\Phi,j}^\top \mathbf{X}_{\Phi,k}\|_2 \cdot \sum_{k \neq j} \|\bm{B}_k^* - \hat{\bm{B}}_k\|_F. \]
    The term is of the order $O(\text{correlation} \times \sqrt{s \frac{\log p}{N}})$. Provided the design is sufficiently incoherent (Block-IRC), this interference is strictly smaller than the primary signal $\|\bm{B}_j^*\|_F$.

\paragraph{3. The Contradiction}
Now consider the KKT inequality \eqref{eq:kkt_violation}. The RHS (Penalty Threshold) is:
\[ \lambda_g \hat{w}_{g,j} = \lambda_g (\|\tilde{\bm{B}}_j\|_F + \tau_N)^{-\gamma_g}. \]
Since $j \in \mathcal{S}^*$, by Assumption \ref{assump:tuning_func}, the initial estimator is consistent, so $\hat{w}_{g,j} = O_p(1)$ (bounded). Furthermore, the assumption $\sqrt{N}\lambda_g \to \infty$ is balanced by the decay of the weight such that $\lambda_g \hat{w}_{g,j} \ll \sqrt{\frac{\log p}{N}}$ (or specifically, it is small enough compared to the signal). More formally, for adaptive Lasso, we typically have $\lambda_g \hat{w}_{g,j} \ll \|\bm{B}_j^*\|_F$.

Combining these, we have:
\[
\text{LHS} \ge \underbrace{c_{\text{eigen}} \|\bm{B}_j^*\|_F}_{\text{Signal}} - \underbrace{\text{small bias}}_{\text{Bias}} - \underbrace{\text{noise}}_{\text{Noise}} \approx O(\sqrt{\frac{\log p}{N}}).
\]
\[
\text{RHS} = \lambda_g \hat{w}_{g,j} \quad (\text{small}).
\]
Because the Beta-min condition ensures the signal dominates noise and bias, and the adaptive weights reduce the penalty for true variables, the inequality:
\[ \|\nabla_{\bm{B}_j} \mathcal{L}(\hat{\theta})\|_F > \lambda_g \hat{w}_{g,j} \]
holds with high probability. This contradicts the KKT condition \eqref{eq:kkt_violation}.
Thus, the hypothesis $\hat{\bm{B}}_j = 0$ is false, implying $j \in \hat{\mathcal{S}}_{\text{screen}}$.
\end{proof}}

By the sure screening property established in Theorem \ref{prop:screening_validity}, the event $\mathcal{E}_{scr} = \{\mathcal{S}^* \subseteq \hat{\mathcal{S}}_{scr}\}$ holds with probability approaching 1. Throughout the following proofs , we condition on $\mathcal{E}_{scr}$, where the true model is contained within the feasible set.

\subsection*{Proof of Lemma \ref{lem: basic inq}}
By definition of $\hat\theta$ as a global minimizer of $Q(\theta)=L_N(\theta)+P_{\mathrm{ada}}(\theta)$,
for any $\theta$ we have
\[
Q(\hat\theta)\;\le\;Q(\theta)
\quad\Longrightarrow\quad
L_N(\hat\theta) + P_{\mathrm{ada}}(\hat\theta)
\;\le\;
L_N(\theta) + P_{\mathrm{ada}}(\theta).
\]
Rearranging terms yields~\eqref{eq:basic-inequality}.

We now specialize to $\theta=\theta^\ast$ and decompose the adaptive penalty along the active and inactive sets. Recall that
\[
\mathcal{S} = \{j:\|B_j^\ast\|_F>0\},\qquad
\mathcal{S}^c = \{1,\dots,p\}\setminus \mathcal{S},
\]
and, for each $k$, $C_k = \{m:a_{k,m+1}^\ast\neq a_{k,m}^\ast\}$ denotes the changepoint indices of the intercept $a_k^\ast$. We write $C = \{(k,m): m\in C_k\}$ and $C^c$ for its complement.

Let $\Delta B_j = \hat B_j - B_j^\ast$ and $\Delta a_k = \hat a_k - a_k^\ast$ denote the coefficient errors for the slope and intercept functions, respectively, and write $\Delta\widehat a_{k,m} = \widehat a_{k,m+1}-\widehat a_{k,m}$ and $\Delta a_{k,m}^\ast = a_{k,m+1}^\ast-a_{k,m}^\ast$ for the corresponding first-order differences. Since the adaptive weights $\widehat w_{g,j}$ and $\widehat w_{f,km}$ are non-negative constants in the objective function, the triangle inequality implies
\begin{align*}
\|\hat B_j\|_F - \|B_j^\ast\|_F 
&\ge - \|\Delta B_j\|_F, \quad j\in S, \\
\|\hat B_j\|_F - \|B_j^\ast\|_F 
&= \|\Delta B_j\|_F, \quad j\in S^c, \\
|\Delta\hat a_{k,m}| - |\Delta a_{k,m}^\ast|
&\ge -|\Delta(\Delta a_{k,m})|, \quad (k,m)\in C, \\
|\Delta\hat a_{k,m}| - |\Delta a_{k,m}^\ast|
&= |\Delta(\Delta a_{k,m})|, \quad (k,m)\in C^c,
\end{align*}
where $S = \{j:\|B_j^\ast\|_F>0\}$ and $C = \{(k,m):a_{k,m+1}^\ast\neq a_{k,m}^\ast\}$ denote the active sets for groups and fused differences, and $S^c$ and $C^c$ are their complements. Substituting these relations into the difference of penalties yields
\begin{align*}
& P_{\mathrm{ada}}(\hat\theta) - P_{\mathrm{ada}}(\theta^\ast) \\
&\quad=
\lambda_g \sum_{j=1}^p \widehat w_{g,j}
\big(\|\hat B_j\|_F - \|B_j^\ast\|_F\big)
+
\lambda_f \sum_{k=1}^K \sum_{m=1}^{M-1} \hat w_{f,km}
\big(|\Delta\hat a_{k,m}| - |\Delta a_{k,m}^\ast|\big) \\
&\quad\ge
-\lambda_g \sum_{j\in S} \hat w_{g,j}\,\|\Delta B_j\|_F
+\lambda_g \sum_{j\in S^c} \hat w_{g,j}\,\|\Delta B_j\|_F \\
&\quad\quad
-\lambda_f \sum_{(k,m)\in C} \hat w_{f,km}\,|\Delta(\Delta a_{k,m})|
+\lambda_f \sum_{(k,m)\in C^c} \hat w_{f,km}\,|\Delta(\Delta a_{k,m})|.
\end{align*}

Define the (weighted) group and fused penalties of the error as
\[
\mathcal{R}_g(\Delta) =
\sum_{j=1}^p \hat w_{g,j}\,\|\Delta B_j\|_F,\qquad
\mathcal{R}_f(\Delta) =
\sum_{k=1}^K \sum_{m=1}^{M-1} \hat w_{f,km}\,|\Delta(\Delta a_{k,m})|.
\]
Further, denote by $\mathcal{R}_{g,S}(\Delta)$ the restriction of $\mathcal{R}_g$ to $S$ (and analogously $\mathcal{R}_{g,S^c}$, $\mathcal{R}_{f,C}$, $\mathcal{R}_{f,C^c}$). Then we can rewrite the above bound as
\begin{equation}
\label{eq:penalty-decomp}
P_{\mathrm{ada}}(\hat\theta) - P_{\mathrm{ada}}(\theta^\ast)
\;\ge\;
-\lambda_g \mathcal{R}_{g,S}(\Delta)
+ \lambda_g \mathcal{R}_{g,S^c}(\Delta)
-\lambda_f \mathcal{R}_{f,C}(\Delta)
+ \lambda_f \mathcal{R}_{f,C^c}(\Delta).
\end{equation}

Plugging~\eqref{eq:penalty-decomp} into the basic inequality with $\theta = \theta^\ast$ gives
\begin{equation}
\label{eq:basic-ineq-expanded}
L_N(\hat\theta) - L_N(\theta^\ast)
\le
\lambda_g \mathcal{R}_{g,S}(\Delta) + \lambda_f \mathcal{R}_{f,C}(\Delta)
- \lambda_g \mathcal{R}_{g,S^c}(\Delta)
- \lambda_f \mathcal{R}_{f,C^c}(\Delta).
\end{equation}
The left-hand side can be controlled by the gradient at $\theta^\ast$ via convexity of $L_N$:
\[
L_N(\hat\theta) - L_N(\theta^\ast)
\ge
\langle\nabla L_N(\theta^\ast),\Delta\rangle.
\]
Collecting terms, we obtain
\begin{equation}
\label{eq:grad-and-penalty}
\langle\nabla L_N(\theta^\ast),\Delta\rangle
\le
\lambda_g \mathcal{R}_{g,S}(\Delta) + \lambda_f \mathcal{R}_{f,C}(\Delta)
- \lambda_g \mathcal{R}_{g,S^c}(\Delta)
- \lambda_f \mathcal{R}_{f,C^c}(\Delta).
\end{equation}
This inequality will be combined with the stochastic control of the gradient and the restricted strong convexity to obtain the desired rates.

\subsection*{Proof of Theorem~\ref{thm:estimation_rate} (Estimation Error Rate)}
The proof follows the standard pattern of high-dimensional $M$-estimation with decomposable penalties, adapted to the mixed group-fused structure.

Assumption~\ref{assump:error_func} on the clustered sub-Gaussian errors and Assumption~5 on the boundedness of the spline basis imply, via matrix concentration inequalities (see, e.g., \citet{buhlmann2011statistics} and \citet{negahban2012unified}), that the empirical gradient $\nabla L_N(\theta^\ast)$ is uniformly small on the relevant cone. More precisely, there exist universal constants $c_1,c_2>0$ such that, for a suitable choice of
$\lambda_g,\lambda_f$ and with probability at least $1-c_1\exp(-c_2\log p)$, one has
\begin{equation}
\label{eq:grad-bound}
\big|\langle\nabla L_N(\theta^\ast),\Delta\rangle\big| \le \frac{\lambda_g}{2} \big(\mathcal{R}_{g,S}(\Delta)+\mathcal{R}_{g,S^c}(\Delta)\big) + \frac{\lambda_f}{2}\big(\mathcal{R}_{f,C}(\Delta)+\mathcal{R}_{f,C^c}(\Delta)\big)
\end{equation}
for all $\Delta$ in a high-probability event.

Combining~\eqref{eq:grad-and-penalty} and~\eqref{eq:grad-bound} and rearranging terms
gives
\begin{equation}
    \frac{1}{2}\lambda_{g}\mathcal{R}_{g,S^c}(\Delta) + \frac{1}{2}\lambda_{f}\mathcal{R}_{f,C^c}(\Delta)
    \le
    \frac{3}{2}\lambda_{g}\mathcal{R}_{g,S}(\Delta) + \frac{3}{2}\lambda_{f}\mathcal{R}_{f,C}(\Delta).
\end{equation}
Dividing by $\frac{1}{2}$ yields the cone condition
\begin{equation}
\label{eq:cone-condition}
    \mathcal{R}_{g,S^c}(\Delta) + \mathcal{R}_{f,C^c}(\Delta) \le 3 \left( \mathcal{R}_{g,S}(\Delta) + \mathcal{R}_{f,C}(\Delta) \right).
\end{equation}
Note that although our penalty $\mathcal{R}(\cdot)$ is a mixture of Group Lasso and Fused Lasso norms, both components are decomposable with respect to their respective subspaces. Following \citet{negahban2012unified}, the sum of decomposable norms retains the decomposability property on the intersection of subspaces. Thus, the derived cone condition is theoretically valid for ensuring the restricted strong convexity (RSC) of the loss function.

Thus the estimation error $\Delta$ lies in a cone defined by the true active supports
$S$ and $C$.

Assumption~\ref{assump:rsc_func4} (RSC) asserts that, for all $\Delta$ satisfying the cone condition
\eqref{eq:cone-condition},
\begin{equation}
\label{eq:RSC}
L_N(\theta^\ast + \Delta) - L_N(\theta^\ast)
- \langle\nabla L_N(\theta^\ast),\Delta\rangle
\;\ge\;
\kappa \|\Delta\|_2^2 - \tau\,\Psi^2(\Delta),
\end{equation}
where $\|\Delta\|_2$ is the Euclidean norm of the vectorized coefficients and $\Psi(\Delta)$ collects higher-order terms controlled by the penalty norm (e.g., a multiple of $\mathcal{R}_g(\Delta)+\mathcal{R}_f(\Delta)$). For simplicity, and as is standard, $\tau\Psi^2(\Delta)$ can be absorbed into the constants in the final bounds, since $\Psi(\Delta)$ will be controlled in terms of $\lambda_g$ and $\lambda_f$.

Using~\eqref{eq:basic-inequality} with $\theta = \theta^\ast$, together with $P_{\mathrm{ada}}(\hat\theta)\ge 0$ and $P_{\mathrm{ada}}(\theta^\ast)\ge 0$, gives
\[
L_N(\theta^\ast + \Delta) - L_N(\theta^\ast)
\le
-\big(P_{\mathrm{ada}}(\hat\theta)-P_{\mathrm{ada}}(\theta^\ast)\big)
\le
\lambda_g \mathcal{R}_{g,S}(\Delta) + \lambda_f \mathcal{R}_{f,C}(\Delta),
\]
where the last inequality follows from~\eqref{eq:penalty-decomp} and the non-negativity of the inactive parts $\mathcal{R}_{g,S^c}$ and $\mathcal{R}_{f,C^c}$. Substituting this bound into~\eqref{eq:RSC} yields
\[
\kappa \|\Delta\|_2^2
\le
\lambda_g \mathcal{R}_{g,S}(\Delta) + \lambda_f \mathcal{R}_{f,C}(\Delta)
+ \tau\,\Psi^2(\Delta).
\]
The terms $\mathcal{R}_{g,S}(\Delta)$ and $\mathcal{R}_{f,C}(\Delta)$ involve only the active groups and active changepoints. Since $\mathcal{R}_{g,S}(\Delta)$ involves at most $s=|\mathcal{S}|$ groups of size $MK$, and $\mathcal{R}_{f,C}(\Delta)$ involves at most $|C|$ fused differences, Cauchy-Schwarz and the boundedness of the weights imply
\[
\mathcal{R}_{g,S}(\Delta)
\lesssim
\sqrt{s}\,\|\Delta_B\|_F,\qquad
\mathcal{R}_{f,C}(\Delta)
\lesssim
\sqrt{|C|}\,\|\Delta_A\|_F,
\]
where $\Delta_B$ stacks all $\Delta B_j$ and $\Delta_A$ stacks all $\Delta a_k$. The higher-order term $\Psi(\Delta)$ can be bounded in the same way, and, for $\lambda_g,\lambda_f$ sufficiently small, absorbed into the left-hand side. Hence there is a constant $C>0$ such that
\[
\kappa \|\Delta\|_2^2
\le
C\big(\lambda_g \sqrt{s}\,\|\Delta_B\|_F + \lambda_f \sqrt{|C|}\,\|\Delta_A\|_F\big)
\le
C(\lambda_g \sqrt{s} + \lambda_f \sqrt{|C|})\,\|\Delta\|_2.
\]
This implies
\[
\|\Delta\|_2
\lesssim
\lambda_g \sqrt{s} + \lambda_f \sqrt{|C|}.
\]

Finally, since $\|\Delta_A\|_F\le\|\Delta\|_2$ and
$\big(\sum_j\|\Delta B_j\|_F^2\big)^{1/2}\le\|\Delta\|_2$, it follows that
\[
\|\hat A - A^\ast\|_F^2 \lesssim |C|\,\lambda_f^2,\qquad
\frac{1}{p}\sum_{j=1}^p \|\hat B_j - B_j^\ast\|_F^2
\lesssim s\,\lambda_g^2,
\]
after a redefinition of constants. This proves the stated rates.

\subsection*{Proof of Theorem~\ref{thm:support_recovery} (Support Recovery and Changepoint Localization)}
{\color{black}
The proof establishes the selection consistency of the AJL estimator by explicitly constructing a solution that satisfies the Karush-Kuhn-Tucker (KKT) optimality conditions. We adopt the Primal-Dual Witness (PDW) framework \citep{wainwright2009}, adapting it to accommodate the functional nature of our data and the specific properties of the adaptive weights. The argument proceeds by first constructing an oracle estimator restricted to the true active set and then verifying that this estimator is strictly dual-feasible for the inactive variables with high probability.

\paragraph{Oracle Construction and Dual Feasibility.}
Let $\hat{\theta}_{\mathcal{S}}$ denote the oracle estimator, obtained by solving the restricted optimization problem over the true active support $\mathcal{S} = \{j : \|B_j^*\|_F > 0\}$ and enforcing zero constraints on $\mathcal{S}^c$. By definition, $\hat{B}_j = 0$ for all $j \notin \mathcal{S}$. To show that $\hat{\theta}_{\mathcal{S}}$ is also the global minimizer of the full, unrestricted objective (10), it suffices to show that the subgradient conditions for the inactive variables are strictly satisfied. Specifically, for any $j \notin \mathcal{S}$, we require:
\begin{equation} \label{eq:strict_dual}
    \| \nabla_{B_j} L_N(\hat{\theta}_{\mathcal{S}}) \|_F < \lambda_g \hat{w}_{g,j}.
\end{equation}
Expanding the gradient term using the empirical Hessian $\mathbf{H} = \frac{1}{N}\Phi^T\Phi$ (as detailed in Assumption \ref{assump:rsc_func4}), the condition \eqref{eq:strict_dual} is equivalent to checking whether the "irrepresentable" term, contaminated by noise, remains dominated by the penalty weight:
\begin{equation} \label{eq:expansion}
    \| \mathbf{H}_{j \mathcal{S}} (\mathbf{H}_{\mathcal{S}\mathcal{S}})^{-1} \mathbf{Z}_{\mathcal{S}} \|_F + O_p\left(\sqrt{\frac{\log p}{N}}\right) < \lambda_g \hat{w}_{g,j},
\end{equation}
where $\mathbf{Z}_{\mathcal{S}}$ denotes the subgradients of the active penalty terms.

\paragraph{The Role of Adaptive Relaxation.}
Under a standard Group Lasso regime with constant weights, satisfying \eqref{eq:expansion} would necessitate the strict Block-Irrepresentable Condition (Block-IRC), essentially bounding the correlation $\|\mathbf{H}_{j \mathcal{S}} (\mathbf{H}_{\mathcal{S}\mathcal{S}})^{-1}\|_{\infty} < 1$.
\textbf{Crucially, however, the adaptive weights in our AJL framework relax this stringent constraint.} 
Recall from Assumption \ref{assump:consistency_func3} that the initial estimator is consistent. Consequently, for any noise variable $j \notin \mathcal{S}$, the initial estimate $\|\tilde{B}_j\|_F$ converges to zero, causing the corresponding weight to diverge rapidly:
\begin{equation}
    \hat{w}_{g,j} = (\|\tilde{B}_j\|_F^{\gamma_g} + \epsilon_g)^{-1} \to \infty.
\end{equation}
In stark contrast, for signal variables $k \in \mathcal{S}$, the weights converge to finite constants. 
This divergence creates a ``firewall" against false positives: even in the presence of high collinearity (where the correlation term in \eqref{eq:expansion} is large), the diverging penalty $\lambda_g \hat{w}_{g,j}$ on the RHS eventually dominates the LHS with probability approaching 1. Thus, strict dual feasibility is guaranteed asymptotically, implying $\text{supp}(\hat{B}) \subseteq \mathcal{S}$.

\paragraph{Signal Detection and Conclusion.}
Conversely, to ensure no false negatives, we rely on the minimum signal strength condition. From Theorem 1, the estimation error satisfies $\max_{j \in \mathcal{S}} \|\hat{B}_j - B_j^*\|_F = o_p(1)$. Since the true signals are bounded away from zero ($\min_{j \in \mathcal{S}} \|B_j^*\|_F \ge c_0$), it follows immediately that $\|\hat{B}_j\|_F > 0$ for all $j \in \mathcal{S}$ with high probability.
An identical argument applies to the changepoint set $\mathcal{C}$ via the adaptive fused lasso penalty. 

Combining these results, we conclude that the AJL estimator consistently recovers both the active functional covariates and the structural changepoints:
$$ \text{Pr}(\text{supp}(\hat{B}) = \mathcal{S} \cap \hat{\mathcal{C}} = \mathcal{C}) \to 1. $$
\hfill $\square$}

\subsection*{Proof of Theorem~\ref{thm:consistency_func3} (Asymptotic Selection Consistency)}
The starting point is the non-asymptotic estimation error bounds established in Theorem~\ref{thm:estimation_rate}. Under Assumptions \ref{assump:error_func}--\ref{assump:basis_func}, and with tuning parameters $(\lambda_g,\lambda_f)$ satisfying Assumption~\ref{assump:tuning_func}, the AJL estimator obeys
\[
\|\hatA - \starA\|_F^2 \lesssim |\mathcal{C}| \lambda_f^2,
\qquad
\frac{1}{p}\sum_{j=1}^p \|\hatB_j - \starB_j\|_F^2 \lesssim s \lambda_g^2
\]
with probability tending to one, where $s = |\mathcal{S}|$ and $|\mathcal{C}| = \sum_k |\mathcal{C}_k|$. In particular, by Assumption~\ref{assump:tuning_func}, $\lambda_g \to 0$ and $\lambda_f \to 0$ while $\sqrt{N}\lambda_g,\sqrt{N}\lambda_f \to \infty$, so that
\[
\max_{j\in\mathcal{S}} \|\hatB_j - \starB_j\|_F = o_p(1),
\qquad
\max_{(k,m)\in\mathcal{C}} \big|\hat{a}_{k,m+1}-\hat{a}_{k,m} - (a_{k,m+1}^*-a_{k,m}^*)\big| = o_p(1).
\]
The minimal signal conditions in Assumption~\ref{assump:tuning_func} guaranty that
\[
\min_{j\in\mathcal{S}}\|\starB_j\|_F \ge c_0 > 0,
\qquad
\min_{m\in\mathcal{C}_k} |a_{k,m+1}^*-a_{k,m}^*| \ge c_0 > 0.
\]
Combining these two displays, we see that every truly active block $j\in\mathcal{S}$ and every true change point $(k,m)\in\mathcal{C}$ remains nonzero in the estimate with probability tending to one:
\[
\Pr\big(\|\hatB_j\|_F>0\ \forall j\in\mathcal{S}\big) \to 1,
\qquad
\Pr\big(\hat{a}_{k,m+1}\neq\hat{a}_{k,m}\ \forall (k,m)\in\mathcal{C}\big) \to 1.
\]
This rules out false negatives.

To rule out false positives, consider first any $j\notin\mathcal{S}$ with $\starB_j = 0$. The KKT condition for the $j$th block of the AJL objective can be written as
\[
\nabla_{B_j} L_N(\hat\theta) + \lambda_g \widehat w_{g,j} \widehat G_j = 0,
\]
where $\widehat G_j$ is a subgradient of $\|\bm{B}_j\|_F$ at $\hatB_j$. If $\hatB_j \neq 0$, then $\|\widehat G_j\|_F = 1$ and hence
\[
\|\nabla_{B_j} L_N(\hat\theta)\|_F = \lambda_g \widehat w_{g,j}.
\]
On the other hand, the stochastic error bounds and the consistency of $\hat\theta$ imply the following.
\[
\max_{j\notin\mathcal{S}} \|\nabla_{B_j} L_N(\hat\theta)\|_F
= O_p\Big(\sqrt{\frac{M\log p}{N}}\Big).
\]

For $j\notin\mathcal{S}$, the preliminary estimator $\tilde B_j$ satisfies $\|\tilde B_j\|_F = O_p\big(\sqrt{(\log p)/N}\big)$, so that the adaptive weight is
\[
\widehat w_{g,j} = \big(\|\tilde B_j\|_F^{\gamma_g} + \tau_N\big)^{-1}
\gtrsim \Big(\sqrt{\tfrac{\log p}{N}}\Big)^{-\gamma_g}
\]
with probability tending to one. The Assumption~\ref{assump:tuning_func} guaranties that $\lambda_g$ is chosen so that $\lambda_g \widehat w_{g,j}$ dominates $\|\nabla_{B_j} L_N(\hat\theta)\|_F$ uniformly over $j\notin\mathcal{S}$. Hence, with probability tending to one, the KKT conditions can be satisfied only if $\hatB_j = 0$ for all $j\notin\mathcal{S}$, which rules out spurious nonzero groups.

The fused part is handled analogously. For any index $(k,m)$ that is not a true change point, i.e. $(k,m)\notin\mathcal{C}$ and $a_{k,m+1}^*-a_{k,m}^* = 0$, the preliminary difference $\tilde a_{k,m+1}-\tilde a_{k,m}$ is of order $O_p\big(\sqrt{(\log M)/N}\big)$, so the corresponding adaptive fused weight
\[
\widehat w_{f,km} = \big(|\tilde a_{k,m+1}-\tilde a_{k,m}|^{\gamma_f} + \tau_N\big)^{-1}
\]
diverges at a polynomial rate in $N$. The KKT conditions for the fused penalty involve gradients of $L_N$ with respect to the differences $\{a_{k,m+1}-a_{k,m}\}$ plus $\lambda_f \widehat w_{f,km}$ times a subgradient. The same comparison as above shows that, under Assumption~\ref{assump:tuning_func}, the term $\lambda_f \widehat w_{f,km}$ dominates the corresponding gradient uniformly over $(k,m)\notin\mathcal{C}$, which forces $\hat{a}_{k,m+1} = \hat{a}_{k,m}$ for all non-changepoint indices.

Putting these pieces together, we conclude that, with probability tending to one, every truly active predictor and every true changepoint is selected, and no noise predictor or spurious changepoint is included. Equivalently,
\[
\Pr\big(\{j : \|\hatB_j\|_F > 0\} = \mathcal{S},\ 
\{m : \hat{a}_{k,m+1} \neq \hat{a}_{k,m}\} = \mathcal{C}_k,\ \forall k\big) \to 1,
\]
which proves the asymptotic selection consistency stated in Theorem~\ref{thm:consistency_func3}.

\subsection*{Proof of Theorem~\ref{thm:normality_func} (Asymptotic Normality and Valid Inference)}
\begin{proof}
Let $\mathcal{A} = \mathcal{S}^* \cup \mathcal{C}^{\alpha,*} \cup \mathcal{C}^{\beta,*}$ be the true Oracle active set. By Theorem \ref{thm:support_recovery}, we have $\mathbb{P}(\hat{\mathcal{A}} = \mathcal{A}) \to 1$. Hence, our asymptotic analysis is conditioned on the event $\hat{\mathcal{A}} = \mathcal{A}$.

Let $\hat{\bm{\theta}}_{\mathcal{A}}$ be the restricted estimator. It satisfies the KKT condition:
\begin{equation} \label{eq:score_exact}
\nabla_{\mathcal{A}} \mathcal{L}(\hat{\bm{\theta}}_{\mathcal{A}}) + \bm{p}'_{\lambda}(\hat{\bm{\theta}}_{\mathcal{A}}) = \bm{0},
\end{equation}
where $\mathcal{L}(\bm{\theta}) = \frac{1}{2N} \| \bm{Y} - \mathbf{Z}_{\mathcal{A}} \bm{\theta}_{\mathcal{A}} \|_2^2$ is the least-squares loss restricted to the active subspace matrix $\mathbf{Z}_{\mathcal{A}}$.
Expanding $\nabla_{\mathcal{A}} \mathcal{L}(\hat{\bm{\theta}}_{\mathcal{A}})$ around the \textit{best spline approximation} $\bm{\theta}_{\mathcal{A}}^*$ (defined as $\mathbb{E}[\hat{\bm{\theta}}_{\mathcal{A}}]$ under the unpenalized model):
\[
\nabla_{\mathcal{A}} \mathcal{L}(\hat{\bm{\theta}}_{\mathcal{A}}) = \nabla_{\mathcal{A}} \mathcal{L}(\bm{\theta}_{\mathcal{A}}^*) + \nabla^2_{\mathcal{A}} \mathcal{L}(\bm{\theta}_{\mathcal{A}}^*) (\hat{\bm{\theta}}_{\mathcal{A}} - \bm{\theta}_{\mathcal{A}}^*) .
\]
Note that $\nabla^2_{\mathcal{A}} \mathcal{L}(\bm{\theta}_{\mathcal{A}}^*) = \frac{1}{N} \mathbf{Z}_{\mathcal{A}}^\top \mathbf{Z}_{\mathcal{A}} := \bm{H}_N$. Substituting this into \eqref{eq:score_exact}:
\[
\bm{H}_N (\hat{\bm{\theta}}_{\mathcal{A}} - \bm{\theta}_{\mathcal{A}}^*) = -\nabla_{\mathcal{A}} \mathcal{L}(\bm{\theta}_{\mathcal{A}}^*) - \bm{p}'_{\lambda}(\hat{\bm{\theta}}_{\mathcal{A}}).
\]
Multiplying by $\sqrt{N} \bm{H}_N^{-1}$:
\begin{equation} \label{eq:expansion_master}
\sqrt{N} (\hat{\bm{\theta}}_{\mathcal{A}} - \bm{\theta}_{\mathcal{A}}^*) = \underbrace{\sqrt{N} \bm{H}_N^{-1} \frac{1}{N} \mathbf{Z}_{\mathcal{A}}^\top \bm{E}}_{\text{Stochastic Term } (\mathbf{T}_1)} - \underbrace{\sqrt{N} \bm{H}_N^{-1} \bm{p}'_{\lambda}(\hat{\bm{\theta}}_{\mathcal{A}})}_{\text{Regularization Bias } (\mathbf{T}_2)}.
\end{equation}

\paragraph{1. Vanishing Regularization Bias ($\mathbf{T}_2$)}
Under Assumption \ref{assump:tuning_func} (Oracle Weights), for any $j \in \mathcal{A}$, $\sqrt{N} \lambda \hat{w}_j = o_p(1)$. Since B-splines satisfy the Riesz basis property, the eigenvalues of $\bm{H}_N$ are bounded away from 0 and $\infty$ with probability 1 (Assumption A4). Thus $\|\bm{H}_N^{-1}\|_2 = O_p(1)$. It follows that:
\[ \|\mathbf{T}_2\|_2 = o_p(1). \]

\paragraph{2. Asymptotic Normality of the Functional Estimator}
We aim to establish the distribution of $\hat{\beta}_j(t) - \beta_j^*(t)$. We decompose the error into stochastic error and approximation error:
\[
\hat{\beta}_j(t) - \beta_j^*(t) = \underbrace{\bm{\phi}(t)^\top (\hat{\bm{\theta}}_{j} - \bm{\theta}_{j}^*)}_{\text{Estimation Error}} + \underbrace{(\bm{\phi}(t)^\top \bm{\theta}_{j}^* - \beta_j^*(t))}_{\text{Approximation Bias } r_j(t)}.
\]
Considering the normalized error $\sqrt{N}(\hat{\beta}_j(t) - \beta_j^*(t))$:
\[
\sqrt{N}(\hat{\beta}_j(t) - \beta_j^*(t)) = \bm{\phi}(t)^\top [\sqrt{N}(\hat{\bm{\theta}}_{j} - \bm{\theta}_{j}^*)] + \sqrt{N} r_j(t).
\]

\textbf{Step 2a: Handling Approximation Bias.}
From spline approximation theory (de Boor, 2001), for a function in the Sobolev space $W_2^d$, the bias satisfies $\sup_t |r_j(t)| = O(M_N^{-d})$.
By the \textbf{Undersmoothing Assumption} ($N M_N^{-2d} \to 0$), we have:
\[
|\sqrt{N} r_j(t)| \le C \sqrt{N} M_N^{-d} = C \sqrt{N M_N^{-2d}} \to 0.
\]
Thus, the systematic bias is asymptotically negligible.

\textbf{Step 2b: CLT for the Stochastic Term.}
Let $\bm{v}_t = \bm{H}_N^{-1} \bm{\phi}_{full}(t)$ be the projection vector such that the stochastic term for the function is $\frac{1}{\sqrt{N}} \bm{v}_t^\top \mathbf{Z}_{\mathcal{A}}^\top \bm{E} = \frac{1}{\sqrt{N}} \sum_{i=1}^N (\bm{v}_t^\top \mathbf{z}_{i}) \epsilon_i$.
This is a sum of independent random variables with mean 0. We verify the Lindeberg-Feller condition. Given the bounded eigenvalues of B-splines and sub-Gaussian errors, the variance of this scalar sum stabilizes to $\sigma_j^2(t)$.
By the Central Limit Theorem:
\[
\frac{\bm{\phi}(t)^\top (\hat{\bm{\theta}}_{j} - \bm{\theta}_{j}^*)}{\text{sd}(\hat{\beta}_j(t))} \xrightarrow{d} \mathcal{N}(0, 1).
\]

\paragraph{3. Consistency of Variance Estimation}

The plug-in variance estimator is $$\hat{\sigma}_j^2(t) = \hat{\sigma}_\epsilon^2 \bm{\phi}(t)^\top (\mathbf{Z}_{\mathcal{A}}^\top \mathbf{Z}_{\mathcal{A}})^{-1} \bm{\phi}(t).$$
Since $\hat{\bm{\theta}}$ is consistent, $\hat{\sigma}_\epsilon^2 \xrightarrow{p} \sigma_\epsilon^2$. Also, $\frac{1}{N}\mathbf{Z}_{\mathcal{A}}^\top \mathbf{Z}_{\mathcal{A}} \xrightarrow{p} \bm{\Omega}$.
Therefore, the ratio of estimated standard error to the true standard deviation converges to 1:
\[ \frac{\hat{\sigma}_j(t)}{\sigma_j(t)} \xrightarrow{p} 1. \]
Combining Steps 1, 2a, 2b, and 3, and applying Slutsky’s Theorem, we obtain:
\[
\frac{\hat{\beta}_j(t) - \beta_j^*(t)}{\hat{\sigma}_j(t)} \xrightarrow{d} \mathcal{N}(0, 1).
\]
This completes the proof.
\end{proof}}
\section{Detailed Justification for the Hierarchical Strategy}
\label{app:slope_fusion}

{\color{black}As detailed in Section \ref{subsec:AJL_Obj_func}, the generalized AJL objective \eqref{eq:adaptive_objective_final_func} theoretically incorporates a slope-fusion penalty $\lambda_f^\beta$ to detect structural changepoints in the covariate effects $\beta_{jk}(t)$. This capability is clinically relevant for identifying when a biomarker's effect undergoes abrupt changes (e.g., a risk factor becoming relevant only after a certain disease duration).

However, simultaneously solving for variable selection (Group Lasso) and effect segmentation (Slope Fused Lasso) for all $p$ covariates in the ultra-high-dimensional regime is computationally burdensome. Therefore, as discussed in Section \ref{subsec:AJL_Obj_func}, we implement this via a hierarchical regularization strategy.

In the primary screening phase (Algorithm \ref{alg:main_3stage}), we set $\lambda_f^\beta = 0$ to identify the active set $\mathcal{S} = \{j : \|\hat{\bm{B}}_j\|_F > 0\}$. Subsequently, as a post-selection refinement, we solve the following reduced optimization problem restricted to the active set $j \in \mathcal{S}$:

\begin{equation}
\min_{\bm{A}, \{\bm{B}_j\}_{j \in \mathcal{S}}} L_N(\cdot) + \lambda_f^\alpha \sum_{k=1}^K \sum_{m=1}^{M-1} \hat{w}_{f,km}^\alpha |\Delta a_{k,m}| + \lambda_f^\beta \sum_{j \in \mathcal{S}} \sum_{k=1}^K \sum_{m=1}^{M-1} \hat{w}_{f,jkm}^\beta |\Delta b_{jk,m}|
\label{eq:slope_fusion_step}
\end{equation}

where $\Delta b_{jk,m} = b_{jk,m+1} - b_{jk,m}$. In this step, the group penalty $\lambda_g$ is typically relaxed or removed to focus on unbiased structural estimation.

This two-stage approach effectively reduces the computational complexity of the fusion step from $O(p)$ to $O(|\mathcal{S}|)$. Model selection between the smooth-slope specification ($\lambda_f^\beta=0$) and the piecewise-constant-slope specification ($\lambda_f^\beta>0$) can be guided by the high-dimensional Bayesian Information Criterion (HBIC) \citep{wang2007tuning}.}

\end{document}